\newif\ifprocs
\procstrue
\procsfalse

\ifprocs

\else
\documentclass[11pt,USenglish]{article}
\usepackage[margin=1.0in]{geometry}
\usepackage[utf8]{inputenc}
\usepackage[english]{babel}
\fi

\usepackage{mwe}
\usepackage{caption}
\usepackage{xcolor}
\usepackage{algorithm}
\usepackage[noend]{algpseudocode}
\usepackage{authblk}

\usepackage{amsmath}
\usepackage{amsthm}
\usepackage{amssymb}
\usepackage{complexity}
\usepackage{graphicx}
\usepackage{xspace}
\usepackage[export]{adjustbox}
\ifprocs
\else
\usepackage[colorlinks,linkcolor=blue,citecolor=blue,filecolor=blue,urlcolor=blue]{hyperref}
\fi

\ifprocs
\else
\newtheorem{theorem}{Theorem}[section]
\newtheorem{lemma}[theorem]{Lemma}

\newtheorem{definition}[theorem]{Definition}

\newtheorem{observation}[theorem]{Observation}

\newtheorem{corollary}[theorem]{Corollary}
\newtheorem{hypothesis}[theorem]{Hypothesis}
\fi

\theoremstyle{plain}
\newtheorem{claim}[theorem]{Claim}

\makeatletter
\newtheorem*{rep@theorem}{\rep@title}
\newcommand{\newreptheorem}[2]{%
\newenvironment{rep#1}[1]{%
 \def\rep@title{#2 \ref{##1}}%
 \begin{rep@theorem}}%
 {\end{rep@theorem}}}
\makeatother

\newreptheorem{theorem}{Theorem}
\newreptheorem{lemma}{Lemma}

\newcommand{\ProblemName}[1]{\textsf{#1}}
\newcommand{\MF}{\ProblemName{Max-Flow}\xspace}

\newcommand{\APMF}{\ProblemName{All-Pairs Max-Flow}\xspace}
\newcommand{\APR}{\ProblemName{All-Pairs Reachability}\xspace}

\newcommand{\TOV}{\ProblemName{$3$OV}\xspace}

\DeclareMathOperator{\SOL}{SOL}
\DeclareMathOperator{\dist}{dist}

\newcommand\eps{\varepsilon}
\newcommand\tO{\ensuremath{\tilde O}}

\newcommand{\T}{\mathcal{T}}
\newcommand{\TG}{\mathcal{T}^*} % guessed tree

\providecommand{\set}[1]{{\{#1\}}}
\providecommand{\card}[1]{\lvert#1\rvert}

\begin{document}

\author[1]{Amir Abboud}
\author[2]{Robert Krauthgamer%
  \thanks{Work partially supported by ONR Award N00014-18-1-2364, the Israel Science Foundation grant \#1086/18, and a Minerva Foundation grant.
    Part of this work was done while the author was visiting the Simons Institute for the Theory of Computing.
  }}
\author[3]{Ohad Trabelsi%
  \thanks{Work partly done at IBM Almaden.}}
\affil[1]{IBM Almaden Research Center. Email: \texttt{amir.abboud@ibm.com}}
\affil[2]{Weizmann Institute of Science. Email: %\texttt{\{robert.krauthgamer,ohad.trabelsi\}@weizmann.ac.il
		 \texttt{robert.krauthgamer@weizmann.ac.il}}
 \affil[3]{Weizmann Institute of Science. Email: \texttt{ohad.trabelsi@weizmann.ac.il}}

  \title{New Algorithms and Lower Bounds for All-Pairs Max-Flow in Undirected Graphs\thanks{A full version appears at \href{http://arxiv.org/abs/1901.01412}{arXiv:1901.01412}}}
%  \date{}

\maketitle

%\thispagestyle{empty}
%\setcounter{page}{0}

%%%%%%%%%%%%%%% abstract %%%%%%%%%%%%%%%%%%%%%
\begin{abstract}
  We investigate the time-complexity of the \APMF problem:
  Given a graph with $n$ nodes and $m$ edges,
  compute for all pairs of nodes the maximum-flow value between them. 
  If \MF (the version with a given source-sink pair $s,t$) can be solved in time $T(m)$, then an $O(n^2) \cdot T(m)$ is a trivial upper bound.
  But can we do better?

For directed graphs, recent results in fine-grained complexity suggest that this time bound is essentially optimal.
In contrast, for undirected graphs with edge capacities,
a seminal algorithm of Gomory and Hu (1961) runs in much faster time $O(n)\cdot T(m)$. 
Under the plausible assumption that \MF can be solved in near-linear time
$m^{1+o(1)}$, this half-century old algorithm yields an $nm^{1+o(1)}$ bound. 
Several other algorithms have been designed through the years,
including $\tilde{O}(mn)$ time for unit-capacity edges (unconditionally),
but none of them break the $O(mn)$ barrier. 
Meanwhile, no super-linear lower bound was shown for undirected graphs.

We design the first hardness reductions for \APMF in undirected graphs,
giving an essentially optimal lower bound for the \emph{node-capacities}
setting.
For edge capacities, our efforts to prove similar lower bounds have failed, 
but we have discovered a surprising new algorithm
that breaks the $O(mn)$ barrier for graphs with unit-capacity edges!
Assuming $T(m)=m^{1+o(1)}$, our algorithm runs in time $m^{3/2 +o(1)}$ and outputs a cut-equivalent tree (similarly to the Gomory-Hu algorithm). 
Even with current \MF algorithms we improve state-of-the-art
as long as $m=O(n^{5/3-\varepsilon})$.
Finally, we explain the lack of lower bounds by proving a \emph{non-reducibility} result. 
This result is based on a new quasi-linear time $\tO(m)$ \emph{non-deterministic} algorithm for constructing a cut-equivalent tree and may be of independent interest.
\end{abstract}

%%%%%%%%%%%%%%% Introduction %%%%%%%%%%%%%%%%%
\newpage

\section{Introduction}

In the maximum $st$-flow problem (abbreviated \MF),
the goal is to compute the maximum value of a feasible flow 
% or equivalently the minimum cut
between a given pair of nodes $s,t$ (sometimes called \emph{terminals})
in an input graph.%
\footnote{Throughout, we focus on computing the value of the flow
  (rather than an actual flow),
  which is equal to the value of the minimum $st$-cut
  by the famous max-flow/min-cut theorem~\cite{FF56}.
}
Determining the time complexity of this problem is one of the most prominent open questions in fine-grained complexity and algorithms.
The best running time known for directed (or undirected)
graphs with $n$ nodes, $m$ edges, and largest integer capacity $U$ is
$\tO(\min\{m^{10/7}U^{1/7}, m \sqrt{n}\log U\})$ \cite{Madry16,LS14},
where throughout $\tO(f)$ hides logarithmic factors
and stands for $O(f \log^{O(1)} f)$. 
To date, there is no $\Omega(m^{1+\eps})$  lower bound for this problem,
even when utilizing one of the popular conjectures of fine-grained complexity, such as the Strong Exponential-Time Hypothesis (SETH) of~\cite{ImpaSETH}.%
\footnote{SETH asserts that for every fixed $\varepsilon>0$
  there is an integer $k\geq 3$,
  such that kSAT on $n$ variables and $m$ clauses
  cannot be solved in time $2^{(1-\varepsilon)n} m^{O(1)}$.  
}
This gap is regularly debated among experts,
and a common belief is that such a lower bound is not possible,
since a near-linear-time algorithm exists but is not yet known.
%\footnote{
There is also a formal barrier for basing a lower bound for \MF on SETH,
  as it would refute the so-called
  Non-deterministic SETH (NSETH)~\cite{carmosino2016nondeterministic}. 
%}
We will henceforth assume that \MF can be solved in time $m^{1+o(1)}$,
and investigate some of the most important questions that remain open under this favorable assumption.
(None of our results need this assumption; it only serves for highlighting their significance.)
%(That said, our results improve state-of-the-art bounds 
%even without this assumption.)
 
Perhaps the most natural next-step after the $s,t$ version is the ``all-pairs'' version (abbreviated \APMF), 
where the goal is to solve \MF for all pairs of nodes in the graph. 
This multi-terminal problem, dating back to 1960 \cite{Mayeda60,Chien60}, 
is the main focus of our work:
\begin{quotation}
\emph{What is the time complexity of computing \MF between all pairs of nodes?}
\end{quotation}

We will discuss a few natural settings,
e.g., directed vs.\ undirected, or node-capacities vs.\ edge-capacities,
in which the answer to this question may vary. 
A trivial strategy for solving this problem (in any setting)
is to invoke a $T(m)$-time algorithm for the $s,t$ version $O(n^2)$ times,
giving a total time bound of $O(n^2)\cdot T(m)$,
which is $n^2 \cdot m^{1+o(1)}$ under our favorable assumption. 
But one would hope to do much better,
as this all-pairs version arises in countless applications, 
such as a graph-clustering approach for image segmentation \cite{WL93}.

In undirected edge-capacitated graphs,
a seminal paper of Gomory and Hu~\cite{GH61} showed in 1961 how to solve \APMF 
using only $n-1$ calls to a \MF algorithm, rather than $O(n^2)$ calls,
yielding an upper bound $O(n) \cdot T(m)$.
(See also~\cite{Gusfield90} for a different algorithm
where all the $n-1$ calls can be executed on the original graph.) 
This time bound has improved over the years,
following the improvements in algorithms for \MF, 
and under our assumption it would ultimately be $n\cdot m^{1+o(1)}$. 
Even more surprisingly, Gomory and Hu showed that all the $n^2$ answers
can be represented using a single tree,
%now known as a Gomory-Hu tree,
which can be constructed in the same time bound.
Formally, %the Gomory-Hu tree
A \emph{cut-equivalent tree} to a graph $G$ 
is an edge-capacitated tree $T$ on the same set of nodes,
with the property that for every pair of nodes $s,t$,
every minimum $st$-cut in $T$ yields a bipartition of the nodes
which is a minimum $st$-cut in $G$, and of the same value as in $T$.%
\footnote{Notice that a minimum $st$-cut in $T$ consists of a
  single edge that has minimum capacity along the unique $st$-path in $T$,
  and removing this edge disconnects $T$ to two connected components.  
  A \emph{flow-equivalent tree} has the weaker property that
  for every pair of nodes $s,t$,
  the maximum $st$-flow value in $T$ equals that in $G$.
  The key difference is that flow-equivalence maintains only the \emph{values}
  of the flows (and thus also of the corresponding cuts).
}
See also~\cite{GT01} for an experimental study,
and the Encyclopedia of Algorithms~\cite{Panigrahi16} for more background. 
The only algorithm that constructs a cut-equivalent tree
without making $\Omega(n)$ calls to a \MF algorithm
was designed by Bhalgat, Hariharan, Kavitha, and Panigrahi~\cite{BHKP07}.
It runs in time $\tilde{O}(mn)$ in unit-capacity graphs
(or equivalently, if all edges have the same capacity),
and utilizes a tree-packing approach that was developed in~\cite{Cole03,HKP07},
inspired by classical results of~\cite{Gabow95} and~\cite{Edmonds70}. 
However, if \MF can indeed be computed in near-linear time,
then none of the later algorithms beat by a polynomial factor
the time bound $n\cdot m^{1+o(1)}$ of Gomory and Hu's half-century old algorithm.

The time complexity of \APMF becomes higher in settings
where Gomory and Hu's ``tree structure'' \cite{GH61} does not hold. 
For instance, in node-capacitated graphs
(where the flow is constrained at intermediate nodes, rather than edges) flow-equivalent trees are impossible, 
since there could actually exist $\Omega(n^2)$ different maximum-flow values
in a single graph \cite{HL07} (see therein 
also an interesting exposition of certain false claims made earlier).%
\footnote{Granot and Hassin~\cite{GH86} considered
  a related but different notion of minimum $st$-cuts with node capacities,
  where an equivalent tree exists and can be computed.
}
Directed edges make the all-pairs problem even harder;
in fact, in this case node-capacities and edge-capacities are equivalent,
and thus this setting does not admit flow-equivalent trees,
see~\cite{Mayeda62,Jelinek63,HL07}. 
In the last decade, different algorithms were proposed to beat the trivial $O(n^2)\cdot T(m)$ time bound in these harder cases.
The known bound for general graphs is $O(m^\omega)$,
due to Cheung, Lau, and Leung~\cite{CLL13}, 
where $\omega<2.38$ is the matrix multiplication exponent. 
A related version, which is obviously no harder than \APMF,
is to ask (among all pairs of nodes) only for flow values that are at most $k$,
assuming unit node capacities; 
for example, the case $k=1$ is the transitive closure problem (reachability). 
For $k=2$, an $\tO(n^\omega)$-time algorithm was shown in~\cite{G+17}, 
and very recently a similar bound was achieved for all $k=O(1)$~\cite{A+18}.
The aforementioned papers~\cite{CLL13,G+17,A+18}
also present improved algorithms for acyclic graphs (DAGs).
In addition, essentially optimal $\tilde{O}(n^2)$-time algorithms
were found for \APMF in certain graph families,
including small treewidth \cite{ACZ98},
planar graphs \cite{LNSW12}, and
surface-embedded graphs \cite{BENW16}.

The framework of fine-grained complexity has been applied to the all-pairs problem in a few recent papers, although its success has been limited to the directed case.
Abboud, Vassilevska-Williams, and Yu~\cite{AVY15} proved SETH-based lower bounds for some multi-terminal variants of \MF,
such as the single-source all-sinks version, but not all-pairs.
Krauthgamer and Trabelsi~\cite{KT18} proved that \APMF cannot be solved
in time $O(n^{3-\eps})$, for any fixed $\eps>0$, unless SETH is false,
even in the sparse regime $m = n^{1+o(1)}$.
This holds also for unit-capacity graphs, and it essentially settles the complexity of the problem for directed sparse graphs, showing that the $O(n^2 )\cdot T(m)$ upper bound is optimal if one assumes that $T(m) = m^{1+o(1)}$.
Recently, Abboud et al.~\cite{A+18} proved a conditional lower bound that
is even higher for dense graphs,
showing that an $O(n^{\omega+1-\eps})$-time algorithm
would refute the $4$-Clique conjecture.
%\footnote{For so-called combinatorial or practical algorithms, that do not suffer from the inefficiencies of fast matrix multiplication, the conditional lower bound is $n^{4-o(1)}$. It is reasonable to assume that Max-Flow can be solved in near-linear time with a practical algorithm, which would qualify as ``combinatorial".}
However, no non-trivial lower bound is known for undirected graphs.

\subsection{The Challenge of Lower Bounds in Undirected Graphs} 
Let us briefly explain the difficulty in obtaining lower bounds for undirected graphs.
Consider the following folklore reduction from Boolean Matrix Multiplication (BMM) to \APR in directed graphs
(the aforementioned special case of \APMF with $k=1$). 
In BMM the input is two $n\times n$ boolean matrices $P$ and $Q$,
and the goal is to compute the product matrix $R$ given by
\[
  R(a,c) := \vee_{b=1}^n \big( P(a,b) \wedge Q(b,c) \big), 
  \qquad
  \forall a,c \in [n]. 
\]
Computing $R$ can be reduced to \APR as follows.
Construct a graph with three layers $A,B,C$ with $n$ nodes each,
where the edges are directed $A \to B \to C$ and represent the two matrices:
$a \in A$ is connected to $b \in B$ iff $P(a,b) =1$; 
and $b \in B$ is connected to $c \in C$ iff $Q(b,c)=1$.
It is easy to see that $R(a,c)=1$ iff node $a \in A$ can reach node $c \in C$ (via a two-hop path).

This simple reduction shows an $n^{\omega-o(1)}$ lower bound for \APR
in dense directed graphs assuming the BMM conjecture.
Higher lower bounds can be proved by more involved reductions
that utilize the extra power of flow over reachability, 
e.g., an $n^{3-o(1)}$ lower bound in sparse directed graphs assuming SETH \cite{KT18}. 
Nevertheless, this simple reduction illustrates the main difficulty
in adapting such reductions to undirected graphs.

Consider the same construction but with \emph{undirected edges}
(i.e., without the edge orientations). 
The main issue is that paths from $A$ to $C$ can now have more than two hops
-- they can crisscross between two adjacent layers before moving on to the next one.
Indeed, it is easy to construct examples in which the product $R(a,c)=0$ but there is a path from $a$ to $c$ (with more than two hops).
Even if we try to use the extra power of flow, giving us information about the number of paths rather than just the existence of a path, it is still unclear how to distinguish flow that uses a two-hop path (YES case) from flow that uses only longer paths (NO case).

A main technical novelty of this work is a trick to overcome this issue. The high-level idea is to design large gaps between the capacities of nodes in different layers in order to incentivize flow to move to the ``next layer''.
Let us exhibit how this trick applies to the simple reduction above. 
Remove the edge orientations from our three-layer graph,
and introduce node capacities,
letting all nodes in $B$, the middle layer, have capacity $2n$,
and all nodes in $A \cup C$, the other two layers, have capacity $1$.
Now, consider the maximum flow from $a \in A$ to $c \in C$.
If $R(a,c)=1$ then there is a two-hop path through some $b \in B$,
which can carry $2n$ units of flow,
hence the maximum-flow value is at least $2n$.
On the other hand, if $R(a,c)=0$ then every path from $a$ to $c$
must have at least four hops, and a maximum flow must be composed of such paths.
Any such path must pass through at least one node in
$A \cup C\setminus\{a,c\}$, whose capacity is only $1$, 
hence the maximum flow is bounded by $|A \cup C \setminus \{ a,c \}| = 2n-2$.
This proves the same $n^{\omega-o(1)}$ lower bound as before,
but now for undirected graphs with node capacities.%
\footnote{The argument can be simplified a bit if we allow nodes of capacity $0$. We also remark that restricting the flow to obey the capacities of the source and the sink makes the problem much easier; this is the version considered by Granot and Hassin \cite{GH86} and mentioned in the previous footnote.}
In Section~\ref{sec:lower} we utilize this trick in a more elaborate way to prove stronger lower bounds.

\subsection{Our Results}

Our main negative result is the first (conditional) lower bound for All-Pairs Max-Flow that holds in undirected graphs.
For sparse, node-capacitated graphs we are able to match the lower bound $n^{3-o(1)}$ that was previously known only for directed graphs~\cite{KT18},
and it also matches the hypothetical upper bound $n^{3+o(1)}$.

\begin{theorem}
\label{thm1}
Assuming SETH, no algorithm can solve \APMF in undirected graphs on $n$ nodes and $O(n)$ edges with node capacities in $[O(n^2)]$ in time $O(n^{3-\eps})$
for some fixed $\eps>0$.
\end{theorem}

Our lower bound holds even under assumptions that are weaker than SETH (see Section~\ref{sec:lower}), as we reduce from the $3$-Orthogonal-Vectors (\TOV) problem.
At a high level, it combines the trick described above for overcoming the challenge in undirected graphs, with the previous reduction of \cite{KT18} from \TOV to the directed case.
However, both of these ingredients have their own subtleties and fitting them together requires adapting and tweaking them very carefully.
%The final outcome is one of the most complicated constructions of a SETH-based lower bound for a graph problem.

\medskip

Following our Theorem~\ref{thm1}, the largest remaining gap in our understanding of \APMF concerns the most basic and fundamental setting: undirected graphs with edge capacities. What is the time complexity of computing a cut-equivalent tree? The upper bound has essentially been stuck at $n \cdot m^{1+o(1)}$ for more than half a century, while we cannot even rule out a near-linear $m^{1+o(1)}$ running time.
To our great surprise, after a series of failed attempts at proving any lower bound, we have noticed a simple way to design a new algorithm for computing cut-equivalent trees for graphs with unit-capacities, breaking the longstanding $mn$ barrier!

%\begin{theorem}\label{Theorem:Algorithm}
%There is an algorithm that, given an undirected graph $G$ with $n$ nodes and $m$ edges (and unit edge capacities) and parameter $1\le d \le n$, constructs a cut-equivalent tree in time $\tilde{O}( md + T(m,n) \cdot m/d )$ where $T(m,n)$ is the time bound for \MF.
%\end{theorem}
%
\begin{theorem}\label{Theorem:Algorithm}
There is an algorithm that, given an undirected graph $G$ with $n$ nodes and $m$ edges (and unit edge capacities) and parameter $1\le d \le n$, constructs a cut-equivalent tree in time $\tilde{O}( md + \Phi(m,n,d) )$, where $\Phi(m,n,d)=\max\{\sum_{i=1}^{m/d} T(m,n,F_i):F_1,\dots,F_{m/d}\geq 0,\sum_{i=1}^{m/d} F_i\leq 2m \}$ and $T(m,n,F)$ is the time bound for \MF. %and $\sum_{i=1}^{m/d} F_i\leq O(m)$.
\end{theorem}

Using the current bound on $T(m,n,F)$ we achieve running time $\tilde{O}(m^{3/2}n^{1/6})$,
and under the plausible hypothesis that $T(m,n)=m^{1+o(1)}$
our time bound becomes $m^{3/2 + o(1)}$.
In the regime of sparse graphs where $m=\tO(n)$ the previous best algorithm of Bhalgat et al. \cite{BHKP07} had running time $\tilde{O}(n^2)$,
whereas we achieve $\tO(n^{5/3})$, or conditionally $n^{3/2+o(1)}$.
In fact, we improve on their upper bound as long as $m=O(n^{5/3-\eps})$.
Clearly, this also leads to improved bounds for \APMF (with unit edge capacities), for which the best strategy known is to compute the tree
and then extract the answers in time $O(n^2)$.

\medskip

The main open question remains: Can we prove any super-linear lower bounds for the edge capacitated case in undirected graphs? Is there an $m^{1+\eps}$ lower bound under SETH for constructing a cut-equivalent tree?
Perhaps surprisingly, we prove a strong barrier for the possibility of such a result.

We follow the non-reducibility framework of Carmosino et al. \cite{carmosino2016nondeterministic}. 
Intuitively, if problem A is conjectured to remain hard for nondeterministic algorithms while problem B is known to become significantly easier for such algorithms, then we should not expect a reduction from A to B to exist. Such a reduction would allow the nondeterministic speedups for problem B to carry over to A.
To formalize this connection, Carmosino et al. introduce NSETH: the hypothesis that SETH holds against co-nondeterministic algorithms.
NSETH is plausible because it is not clear how a powerful prover could convince a sub-$2^{n}$-time verifier that a given CNF formula is \emph{not} satisfiable.
Moreover, it is known that refuting NSETH requires new techniques since it implies new circuit lower bounds.
Then, they exhibit nondeterministic (and co-nondeterministic) speedups for problems such as $3$-SUM and \MF (using LP duality), showing that a reduction from SAT to these problems would refute NSETH.
  
Our final result builds on Theorem~\ref{Theorem:Algorithm} to design
a quasi-linear time%
\footnote{We say that a time bound $T(n)$ is quasi-linear
  if it is bounded by $O(n\log^c n)$ for some positive constant $c>0$.
}
\emph{nondeterministic} algorithm for constructing a cut-equivalent tree.
This algorithm can perform nondeterministic choices and in the end,
outputs either a correct cut-equivalent tree or ``don't know'' (i.e., aborts),
however we are guaranteed that for every input graph there is a at least one sequence of nondeterministic choices leading to a correct output.
%We design a prover-verifier protocol such that: given a graph $G$ on $m$ edges, the prover (that runs in polynomial time) sends a tree $H$ and a proof of size $\tilde{O}(m)$ that convinces the verifier that $H$ is a cut-equivalent tree of $G$. Most importantly, the runtime of the verifier is $\tilde{O}(m)$.
This result could have applications in computation-delegation settings and may be of interest in other contexts.
In particular, since our nondeterministic witness can be constructed deterministically efficiently, namely, in polynomial but super-linear time, 
it provides a potentially interesting \emph{certifying algorithm} \cite{certify_survey11,alkassar2011verification} (see \cite{Kunnemann18} for a recent paper with a further discussion of the connections to fine-grained complexity).
Our final non-reducibility result is as follows. 

\begin{theorem}
\label{thm3}
If for some $\eps>0$ there is a deterministic fine-grained reduction proving an $\Omega(m^{1+\eps})$ lower bound under SETH for constructing a cut-equivalent tree of an undirected unit edge-capacitated graph on $m$ edges, then NSETH is false.
\end{theorem}

Our result (and this framework for non-reducibility) does not address the possibility of proving a SETH based lower bound with a randomized fine-grained reduction. 
This is because NSETH does not remain plausible when faced against randomization (see \cite{carmosino2016nondeterministic,Wil16_MASETH}).
That said, we are not aware of any examples where this barrier has been successfully bypassed with randomization.

\paragraph{Roadmap.} Our main algorithm is described in the Section~\ref{GHU_Alg}. The nondeterministic algorithm and non-reducibility result are presented in Section~\ref{sec:nondet}. We then present our lower bounds in Section~\ref{sec:lower}. The last section discusses open questions.

\section{Algorithm for a Cut-Equivalent Tree}
\label{GHU_Alg}

The basic strategy in our algorithm for unit edge capacities
is to handle separately nodes whose connectivity (to other nodes)
is high from those whose connectivity is low.
The motivation comes from the simple observation that the degree of a node is an upper bound on the maximum flow from this node to any other node in the graph.
Specifically, our algorithm has two stages.
The first stage uses one method (of partial trees~\cite{HKP07,BHKP07}), 
to compute the parts of the tree that correspond to small connectivities,
and the second stage uses another method
(the classical Gomory-Hu algorithm~\cite{GH61})
to complete it to a cut-equivalent tree (see Figure~\ref{Figs:GH_trees_3}). 
Let us briefly review these two methods.

\paragraph{The Gomory-Hu algorithm.}
This algorithm constructs a cut-equivalent tree $\T$ in iterations. 
Initially, $\T$ is a single node associated with $V$ (the node set of $G$), 
and the execution maintains the invariant that $\T$ is a tree; 
each tree node $i$ is a \emph{super-node},
which means that it is associated with a subset $V_i\subseteq V$; 
and these super-nodes form a partition $V=V_1 \sqcup\cdots\sqcup V_l$.
At each iteration, the algorithm picks arbitrarily two graph nodes $s,t$ 
that lie in the same tree super-node $i$, i.e., $s,t\in V_i$.
The algorithm then constructs from $G$ an auxiliary graph $G'$
by merging nodes that lie in the same connected component of  $\T\setminus\set{i}$
and invokes a \MF algorithm to compute in this $G'$ a minimum $st$-cut, denoted $C'$.
(For example, if the current tree is a path on super-nodes $1,\ldots,l$, 
then $G'$ is obtained from $G$ by merging $V_1\cup\cdots\cup V_{i-1}$
into one node and $V_{i+1}\cup\cdots\cup V_l$ into another node.)
The submodularity of cuts ensures that this cut is also 
a minimum $st$-cut in the original graph $G$,
and it clearly induces a partition $V_i=S\sqcup T$ with $s\in S$ and $t\in T$. The algorithm then modifies $\T$ by splitting super-node $i$
into two super-nodes, one associated with $S$ and one with $T$,
that are connected by an edge whose weight is the value of the cut $C'$,
and further connecting each neighbor of $i$ in $\T$ 
to either $S$ or $T$ (viewed as super-nodes),
depending on its side in the minimum $st$-cut $C'$
(more precisely, neighbor $j$ is connected to the side containing $V_j$).

\begin{figure}[!ht]
%   \centering
       \includegraphics[width=1.0\textwidth,left]{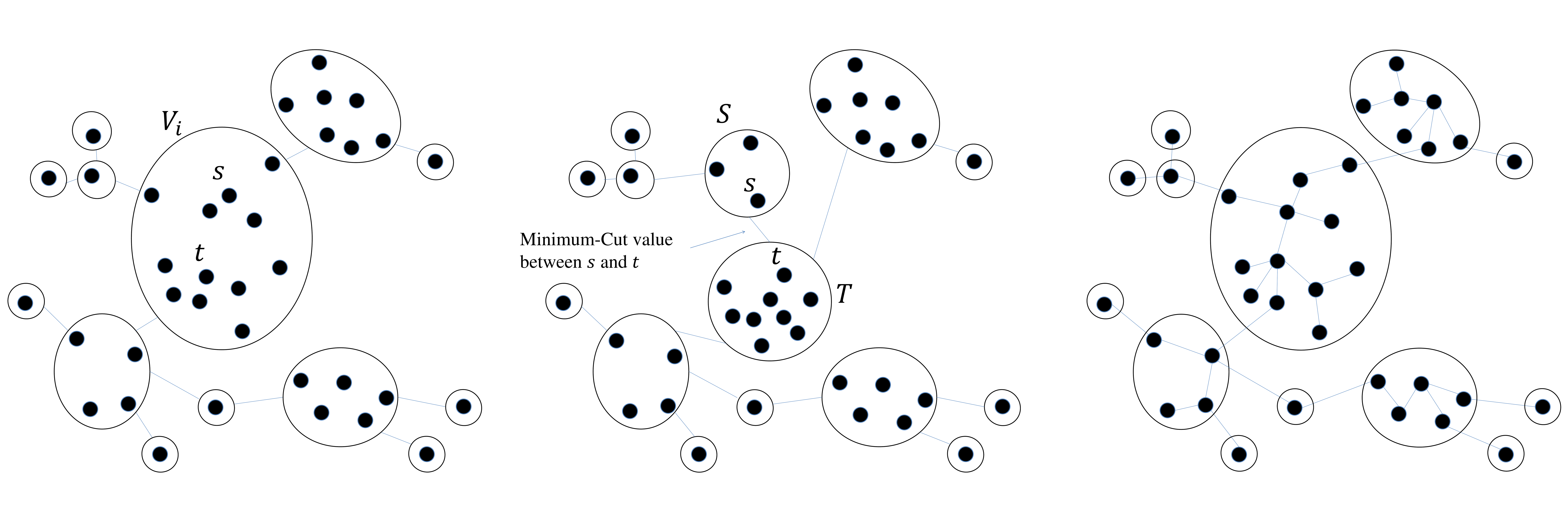}
   \caption[-]{
   An illustration of the construction of $\T$.  Left: $\T$ right before the partition of the super-node $V_i$. Middle: after the partitioning of $V_i$ Right: $\T$ as it unfolds after the Gomory-Hu algorithm finishes.
   }
   \label{Figs:GH_trees_3}
\vspace{.1in}\hrule
\end{figure}

The algorithm performs these iterations until all super-nodes are singletons,
and then $\T$ is a weighted tree with effectively the same node set as $G$.
It can be shown \cite{GH61} that for every $s,t\in V$,
the minimum $st$-cut in $\T$, viewed as a bipartition of $V$,
is also a minimum $st$-cut in $G$, and of the same cut value.
We stress that this property holds regardless of the choice made at each step
of two nodes $s\neq t\in V_i$.

\paragraph{Partial Tree.}
A $k$-partial tree, formally defined below, can also be thought of as
the result of contracting all edges of weight greater than $k$
in a cut-equivalent tree of $G$. 
Such a tree can obviously be constructed using the Gomory-Hu algorithm,
but as stated below (in Lemma~\ref{Lemma:Partial}), 
faster algorithms were designed in~\cite{HKP07,BHKP07},
see also~\cite[Theorem $3$]{Panigrahi16}.
We show below (in Lemma~\ref{lem:truncatedGH}) that such a tree
can be obtained also by a truncated execution of the Gomory-Hu algorithm,
and finally we use this simple but crucial fact to prove our main theorem.

\begin{definition}[$k$-Partial Tree~\cite{HKP07}] 
  A \emph{$k$-partial tree} of a graph $G=(V,E)$ is a tree on $l\leq \card{V}$ super-nodes constituting a partition $V=V_1 \sqcup\cdots\sqcup V_l$, with the following property: 
  For every two nodes $s,t\in V$ whose minimum-cut value in $G$ is at most $k$, let $S,T$ be the super-nodes for which $s\in S$ and $t\in T$,
  then the minimum $ST$-cut in the tree defines a bipartition of $V$
  which is a minimum $st$-cut in $G$ and has the same value.
\end{definition}

\begin{lemma}[\cite{BHKP07}]
\label{Lemma:Partial}
There is an algorithm that given an undirected graph with $n$ nodes and $m$ edges with unit edge capacities and an integer $k\in [n]$, constructs a $k$-partial tree in time $\tO(mk)$.
\end{lemma}

\begin{lemma} \label{lem:truncatedGH}
  Given a $k$-partial tree $T_{low}$ of a graph $G=(V,E)$, there is a truncated execution of the Gomory-Hu algorithm that produces $T_{low}$ (i.e., its auxiliary tree $\T$ becomes $T_{low}$).
  %In particular, it can be completed to a complete cut tree.
\end{lemma}

\begin{proof}
  Consider an execution of the Gomory-Hu algorithm with the following choices.
  At each iteration, pick any two nodes $s,t\in V$
  that lie in the same super-node $i$ of the current tree $\T$
  (hence they are feasible choice in a Gomory-Hu execution)
  \emph{but furthermore} lie in different super-nodes of $T_{low}$,
  as long as such $s,t$ exist. 
  Then split super-node $i$ of $\T$ using the minimum $st$-cut induced by $T_{low}$ (rather than an arbitrary minimum $st$-cut).
  As this cut corresponds to an edge in $T_{low}$,
  it cannot split any super-node of $T_{low}$,
  which implies, by an inductive argument,
  that the super-nodes of $T_{low}$ are subsets of the super-nodes of $\T$,
  and thus our chosen cut is a feasible choice for a Gomory-Hu execution. 
  Notice also that a pair $s,t$ as required above can be chosen
  as long as $\T$ is not equal to $T_{low}$,
  hence the Gomory-Hu execution continues until $\T$ becomes exactly $T_{low}$.
\end{proof}

\medskip
We are now ready to prove our main theorem.

\begin{proof}[Proof of Theorem~\ref{Theorem:Algorithm}]
Let $G=(V,E)$ be an input undirected graph with unit edge capacities,
and denote by $V_{low}$ all the nodes in $G$ whose degrees are at most
the chosen parameter $d\in [n]$,
and by $V_{high}=V\setminus V_{low}$ the nodes whose degrees are greater than $d$.

First use Lemma~\ref{Lemma:Partial} to construct a $d$-partial tree $T_{low}$,
and treat it as the auxiliary tree computed by a truncated execution
of the Gomory-Hu algorithm.
Then continue a Gomory-Hu execution (using this tree)
to complete the construction of a cut-equivalent tree.
Note that every node in $V_{low}$ is in a singleton super-node of $T_{low}$, since its minimum cut value to any other node is at most $d$; thus a super-node $V_i$ in $T_{low}$ has more than one node if and only if it contains only nodes in $V_{high}$.
Moreover, by the properties of $T_{low}$, two nodes have minimum-cut value greater than $d$ if and only if they are in the same super-node $V_i$.
Since by Lemma~\ref{lem:truncatedGH} there exists a truncated Gomory-Hu execution that produces $T_{low}$, a Gomory-Hu execution starting with $T_{low}$ as the auxiliary tree will result in a cut-equivalent tree and the correctness follows.
The running time bound follows as the first step of constructing $T_{low}$ takes $\tO( m d )$ time,
and the second step of the Gomory-Hu execution
takes $\card{V_{high}}$ invocations of \MF, that is running time $\sum_{i=1}^{m/d}T(m,n,F_i)$.
Since every invocation of maximum $st$-flow with value $F_i$ in our algorithm determines a unique edge with capacity $F_i$ in the final cut-equivalent tree, and the sum of the capacities over all the edges of the cut-equivalent tree is at most $2m$ (see Claim~\ref{claim:technical})
%
%\onote{I put here a reference to the next section, where the definitions occur and a bigger claim is proved
%}
%$\sum_{uv\in E_{\T}} c_{\T}(u,v)\leq 2m$, 
it holds for the invocations of \MF that $\sum_{i=1}^{m/d}T(m,n,F_i)\leq 2m$.
Thus, the proof of Theorem~\ref{Theorem:Algorithm} is concluded.

%and the second step of the Gomory-Hu execution
%takes $\card{V_{high}}\cdot MF(m,n)=O(m/d)\cdot MF(m,n)$ time. Thus, the proof of Theorem~\ref{Theorem:Algorithm} is concluded.
\end{proof}

We use the $T(m,n,F)=O(m^{3/4}n^{1/4}F^{1/2})$ time algorithm by~\cite{SidfordT18} to optimize our running time.
By the concavity of $F^{1/2}$, the maximum of $\sum_{i=1}^{m/d}T(m,n,F_i)$ is where always $F_i=d$.
By setting $d=\sqrt{m}n^{1/6}$ we get
$\sum_{i=1}^{\sqrt{m}/n^{1/6}} m^{3/4}n^{1/4}m^{1/4}n^{1/12}=
\sum_{i=1}^{\sqrt{m}/n^{1/6}} m n^{1/3}=
m^{3/2}n^{1/6}
$
which is faster than the currently known $\tO(mn)$ algorithm \cite{BHKP07} whenever $m\in [n,n^{5/3}]$.
Finally, relying on a hypothetical $m^{1+o(1)}$-time algorithm for \MF,
we could set $d=\sqrt{m}$ to get a total running time of $m^{1+o(1)}\cdot m/\sqrt{m}+\tO(m\cdot \sqrt{m})\leq m^{3/2+o(1)}$, as claimed. 

%
%We can now derive the time bounds stated in the Introduction,
%given known \MF algorithms.
%Using the $\tO(m^{10/7})$-time algorithm by~\cite{Madry16}, we can set $d=m^{5/7}$ and get a total running time of $\tO( m^{10/7}\cdot m/(m^{5/7})+ m\cdot m^{5/7} ) = \tO(m^{12/7})$. 
%
%Using the $\tO(m\sqrt{n})$-time algorithm by~\cite{LS14}, we can set $d=\sqrt{m}\cdot n^{1/4}$ and get a total running time of $\tO( m\sqrt{n} \cdot m/(\sqrt{m}\cdot n^{1/4})+ m\cdot \sqrt{m}\cdot n^{1/4})=\tO(m^{3/2}n^{1/4})$.
%Thus, for $m\in [n,n^{7/6}]$ it is faster to use the algorithm in~\cite{Madry16}, while for $m\in [n^{7/6},n^{3/2}]$ it is faster to use the algorithm in~\cite{LS14}. 
%Note that for every $m\in [n,n^{5/3})$ our algorithm is faster than the currently known $\tO(mn)$ algorithm \cite{BHKP07}.
%
%Finally, relying on a hypothetical $m^{1+o(1)}$-time algorithm for \MF,
%we could set $d=\sqrt{m}$ to get a total running time of $m^{1+o(1)}\cdot m/\sqrt{m}+\tO(m\cdot \sqrt{m})\leq m^{3/2+o(1)}$, as claimed. 

\section{Quasi-Linear Nondeterministic Algorithm for Cut-Equivalent Tree}\label{sec:nondet}

As no conditional lower bounds are known
for the problem of constructing a cut-equivalent tree, 
one potentially promising approach is to design a reduction from SAT
to prove that running time $n^{1+\delta-o(1)}$, for a fixed $\delta>0$,
is not possible assuming SETH. 
However, in this section we show that the existence of such a reduction (at least in the case of unit edge-capacities) would refute NSETH. 
This proves our Theorem~\ref{thm3}.

Our main technical result in this section (Theorem~\ref{nondet:main})
is a fast \emph{nondeterministic} algorithm 
for constructing a cut-equivalent tree (the meaning of this notion will be formalized shortly).
We then reach the conclusion about NSETH
by following an argument first made in~\cite{carmosino2016nondeterministic},
however we have to rewrite their argument
(rather than use their definitions and results directly), 
in order to adapt it from decision problems or functions (where each input has exactly one output) to total functions, 
since every graph has at least one cut-equivalent tree (see Section~\ref{subsec:reduction}).

Generally speaking, a search problem $P$ is a binary relation,
and we say that $S$ is a solution to instance $x$ iff $(x,S)\in P$. 
Let $\SOL(x)=\set{S: (x,S)\in P}$ denote the set of solutions for instance $x$. 
We say that $P$ is a \emph{total function}%
\footnote{We use this name for consistency with previous literature,
  although it is really a relation rather than a function.
}
if every instance $x$ has at least one solution, i.e., $\SOL(x)\neq \emptyset$.
Let $\bot$ be the ``don't know'' symbol and assume that $\bot \notin \SOL(x)$ for all $x$.
For example, in our problem of constructing a cut-equivalent tree,
$x$ is a graph and $\SOL(x)$ is the set of all cut-equivalent trees for $x$. 

\begin{definition}[Nondeterministic complexity of a total function]
We say that a total function $P$ has nondeterministic time complexity $T(n)$ if there is a deterministic Turing Machine $M$ such that for every instance $x$ of $P$ with size $\card{x}=n$:
%We say that $M$ is a nondeterministic algorithm for a total function $P$ with running time $T$ if for every instance $X$ of $P$:
%%if given an instance $X$ of $P$ and a guess $g$,
%%its output $M(X,g)$ satisfies:
%%that for every instance $X$ of $P$:
\begin{enumerate} 
\renewcommand{\theenumi}{\alph{enumi}}
\item \label{it:time}
  $\forall g, \DTIME(M(x,g))\leq T(n)$,
  i.e., the time complexity of $M$ is bounded by $T(n)$; 
\item \label{it:a}
  $\exists g, M(x,g)\in \SOL(x)$, i.e., at least one guess leads $M$ to output a solution;
\item \label{it:b}
  $\forall g, M(x,g) \in \{\bot\}\cup \SOL(X)$, i.e., every guess leads $M$ to output either a solution or ``don't know''.
\end{enumerate}
\end{definition}
%\rnote{I think we want/need $\DTIME(M(x,g))\leq T(n)$
%  also in the first requirement (not only in the second one). 
%  But then I realized that a simpler and probably better change
%  would be to have the time bound as a separate requirement.
%  Does it break anything in the proof? Is it conceptually wrong? 
%}

%\rnote{I change TIME to DTIME, to be super-clear it is det time, 
%  and also used a macro from the latex package complexity.
%  Can you please make this change throughout?
%}

We can now state the main technical result of this section. 
We prove it in Section~\ref{sec:nondet_main},
and then use it in Section~\ref{subsec:reduction} to prove Theorem~\ref{thm3}.

\begin{theorem}\label{nondet:main}
The nondeterministic complexity of constructing a cut-equivalent tree
  for an input graph with unit edge-capacities is $\tO(m)$,
  where $m$ is the number edges in the graph. 
\end{theorem}

This algorithm employs the Gomory-Hu algorithm in a very specific manner,
where the vertices chosen at each iteration are ``centroids'' (see below).
The same choice was previously used by Anari and Vazirani~\cite{AnariV18} 
in the context of parallel algorithms (for planar edge-capacitated graphs),
to achieve a logarithmic recursion depth, which is key for parallel time. 
However, since our goal is different (we want near-linear total time) we have to worry about additional issues, besides the depth of the recursion.
Many auxiliary graphs must be handled throughout the execution of the algorithm, and for each one we need to verify multiple minimum cuts. 
This is done by guessing cuts and flows, and the main challenge is to argue that the total size of all these objects (the auxiliary graphs, and the cuts and flows within them) is only $\tO(m)$. 
Towards overcoming this challenge, we show a basic structural result about cut-equivalent trees (see Claim~\ref{claim:technical} below) which may have other applications.
Prior to our work, it seemed unlikely that the Gomory-Hu approach could come close to near-linear time, even if \MF could be computed in linear time, since a \MF computation is executed many times in many auxiliary graphs. 
However, our analysis shows that the total size of all these auxiliary graphs can be near-linear (if the right vertices are chosen at each iteration),
giving hope that this approach may still achieve the desired upper bound.

\subsection{The Nondeterministic Algorithm}
\label{sec:nondet_main}

% \begin{proof}[Proof of Theorem~\ref{nondet:main}]
We now prove Theorem~\ref{nondet:main}.
Let $G=(V,E)$ be the input graph, and let $n=\card{V}$ and $m=\card{E}$.

\paragraph{Overview.}
At a high level, the nondeterministic algorithm 
first guesses nondeterministically a cut-equivalent tree $\TG$, 
and then verifies it by a (nondeterministic) process that resembles
an execution of the Gomory-Hu algorithm that produces $\TG$.
Similarly to the actual Gomory-Hu algorithm, 
our verification process is iterative and maintains a tree $\T$ of super-nodes, 
which means, as described in Section~\ref{GHU_Alg},
that every tree node $i$ is associated with $V_i\subseteq V$, 
and these super-nodes form a partition $V=V_1 \sqcup\cdots\sqcup V_l$. 
This tree $\T$ is initialized to have a single super-node corresponding to $V$
and then modified at each iteration, 
hence we shall call it the \emph{intermediate tree}. 
If all guesses work well, then eventually every super-node is a singleton
and the tree $\T$ corresponds to $\TG$.
Otherwise (some step in the verification fails), the algorithm outputs $\bot$. 

In a true Gomory-Hu execution, every iteration partitions some super-node
into exactly two super-nodes connected by an edge (say $V_i=S\sqcup T$).
In contrast, every iteration of our verification process partitions
some super-node into multiple super-nodes that form a star topology, 
whose center is a singleton
(say $V_i=\set{w}\sqcup V_{i,1}\sqcup\dots\sqcup V_{i,d}$,
where super-node $\set{w}$ has edges to all super-nodes $V_{i,1},\ldots,V_{i,d}$).
We call this an \emph{expansion step} (see Figure~\ref{Figs:Expansion}), 
and the node in the center of the star (i.e., $w$) the \emph{expanded node}.
These expansion steps will be determined from the guess $\TG$. 
For example, in the extreme case that $\TG$ itself is a star,
our verification process will take only one expansion step 
instead of $\card{V}-1$ Gomory-Hu steps.

\begin{figure}[!ht]
%   \centering
       \includegraphics[width=1.0\textwidth,left]{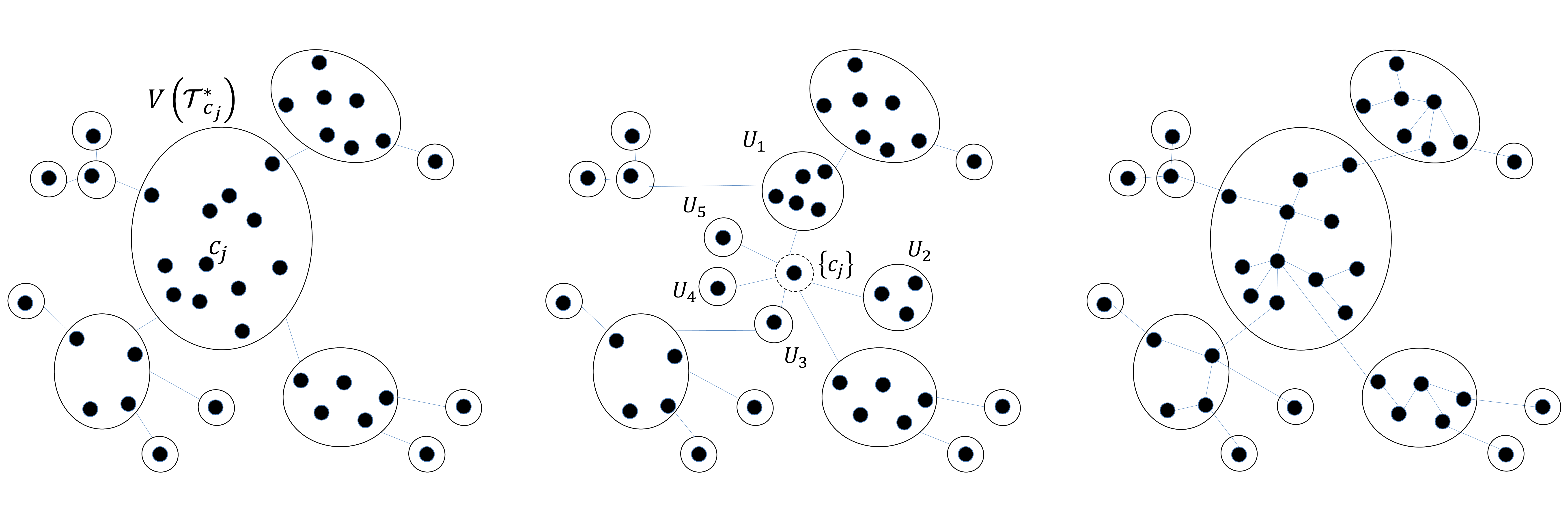}
   \caption[-]{
   An illustration of the verification of a guessed tree $\TG$. Left: the intermediate tree $\T$ right before an expansion step of the node $c_j$ in the super-node $V(\TG_{c_j})$. Middle: after the expansion step (of $c_j$, in the dashed circle) where $U_1,...,U_4$ are $c_j$'s neighbors in $\T^{(j+1)}$ such that $\bigcup_{i=1}^4 U_i\cup \set{c_j}=V(\TG_{c_j})$. Right: the guessed cut-equivalent tree $\TG$.
   }
   \label{Figs:Expansion}
\vspace{.1in}\hrule
\end{figure}

To prove that our algorithm is correct, we will show that every expansion step
corresponds to a valid sequence of steps in the Gomory-Hu algorithm.
As the latter relies on minimum-cut computations in some auxiliary graph $G'$,
also our verification will need minimum-cut computations, 
which can be easily performed in nondeterministic linear time. 
However, this will not achieve overall running time $\tO(m)$, 
because in some scenarios (e.g., in the above example where $\TG$ is a star), 
most of the $\card{V}-1$ minimum-cut computations are performed 
on an auxiliary graph $G'$ of size that is comparable to $G$, i.e., $\Omega(m)$. 
We overcome this obstacle using two ideas. 
First, we compute simultaneously all the minimum-cuts of the same expansion step
in nondeterministic time that is \emph{linear} in the size of $G'$. 
Second, we design a specific sequence of expansion steps 
such that the total size of \emph{all auxiliary graphs $G'$} is $\tO(m)$.

\paragraph{Detailed Algorithm.}

The algorithm first guesses nondeterministically an edge-capacitated tree $\TG$,
and then verifies, as explained below, that it is a cut-equivalent tree. 
Here, verification means that upon the failure of any step,
e.g., verifying some equality (say between the cut and flow values), 
the algorithm terminates with output $\bot$.
(By the same reasoning, we may assume that all guesses are proper,
e.g., a guessed tree is indeed a tree). 
The verification process starts by picking
a sequence of nodes $c_0,c_1,c_2,\ldots$ using the guess $\TG$, 
as follows. 
Recall that a \emph{centroid} of a tree is a node 
whose removal disconnects the tree into connected components (subtrees),
each containing at most half the nodes in the tree.
It is well-known that in every tree, 
a centroid exists and can be found in linear time. 
In a \emph{recursive centroid decomposition} of a tree, 
one finds a centroid of the given tree, removes it
and then repeats the process recursively in every connected component, 
until all remaining components are singletons (have size one). 
Our verification process computes this decomposition for the guess $\TG$, 
which takes time $O(n\log n)$.
For each recursion depth $i\ge 0$ (where clearly $i\le \log n$),
denote the set of centroids computed at depth $i$ by $D_i\subset V$. 
For example, $D_0$ contains exactly one centroid, of the entire $\TG$.
Now let $c_0,c_1,c_2,\ldots$ be the centroids in this decomposition
in order of increasing depth, i.e., starting with the one centroid $c_0\in D_0$,
followed by the centroids from $D_1$ (ordered arbitrarily), and so forth. 
Let $\TG_{c_j}$ be the subtree of $\TG$ in which the centroid $c_j$ was computed; for example $\TG_{c_0}=\TG$.
%The following observation is immediate. 
\begin{observation}\label{observation:disjoint}
For every two centroids from the same depth, namely, $c_j\neq c_{j'}\in D_i$, 
the corresponding subtrees $\TG_{c_j}$ and $\TG_{c_{j'}}$ are node disjoint. 
\end{observation}

The verification process now initializes a tree $\T$,
called the intermediate tree, 
to consist of a single super-node associated with $V$, 
and then performs on it expansion steps for nodes $c_0,c_1,c_2,\ldots$ 
(in this order) as explained below. 

We now explain how to perform an expansion step for node $c_j$.
Recall that $c_j$ is a centroid of the subtree $\TG_{c_j}$,
therefore it defines a partition 
$V(\TG_{c_j}) = \set{c_j}\sqcup U_1 \sqcup\dots\sqcup U_d$,
where $U_1,\ldots,U_d$ are the connected components after removing $c_j$. 
Notice that $d = \deg_{\TG_{c_j}}(c_j)\leq \deg_{\TG}(c_j)$,
and that each $U_k$, $k\in [d]$, contains exactly one node $u_k\in U_k$
that is a neighbor of $c_j$ in $\TG_{c_j}$. 
The expansion step replaces the super-node $V(\TG_{c_j})$ in $\T$ 
with $d+1$ super-nodes $\set{c_j},U_1,\dots,U_d$.
(We slightly abuse notation and use a subset of nodes like $V(\TG_{c_j})$
also to refer to the super-node in $\T$ associated with this subset.) 
These $d+1$ new super-nodes are connected by a star topology,
where the singleton $\set{c_j}$ at the center
and each newly-added edge $(\set{c_j}, U_k)$
is set to the same capacity as the edge $(c_j, u_k)$ in the guess $\TG$. 
In addition, every edge that was incident to super-node $V(\TG_{c_j})$,
say $(V(\TG_{c_j}),W)$, is modified to an edge $(U,W)$, 
where $U$ is one of the new super-nodes $\set{c_j},U_1,\dots,U_d$,
chosen according to the edge in $\TG$ that was used to set
a capacity for $(V(\TG_{c_j}),W)$. 
% (that edge has exactly one node inside $V(\TG_{c_j}) = \set{c_j}\sqcup U_1 \sqcup\dots\sqcup U_d$).
% This completes the modification of the intermediate tree (for this expansion step).
(We will explain how the algorithm verifies the correctness of these edge weights shortly.)

It is easy to verify that the modifications to $\T$ (due to expansion steps)
maintain the following property: 
Every super-node $U$ in $\T$ induces a subtree of $\TG$, 
i.e., the induced subgraph $\TG[U]$ is connected. 
Moreover, eventually every super-node will be a singleton, 
and the intermediate tree will exactly match the guess $\TG$. 
When we need disambiguation, we may use $\T^{(j)}$ 
to denote the tree's state before the expansion step for $c_j$. 
For example, $\T^{(0)}$ is the initial tree with a single super-node $V$.

Informally, the verification algorithm still has to check that the capacities
of the newly-added tree edges correctly represent minimum-cut values. 
%which means they match the values of certain cuts, and that these are minimum cuts.  
To this end, the algorithm now constructs an auxiliary graph $G'_j$ 
just as in the Gomory-Hu algorithm (see Section~\ref{GHU_Alg}).
Specifically, $G'_j$ is constructed by taking $G$,
and then for each connected component of $\T^{(j)}\setminus \set{V(\TG_{c_j})}$
(i.e., after removing super-node $V(\TG_{c_j})$ from $\T^{(j)}$),
merging the nodes in (all the super-nodes in) this component into a single node.
Our analysis shows (in Claim~\ref{claim:GH}) 
%based on the equivalence to a Gomory-Hu execution,
that for all $s,t\in V(\TG_{c_j})$, 
every minimum $st$-cut in the auxiliary graph $G'_j$
is also a minimum $st$-cut in $G$. 
In addition, all the auxiliary graphs of a single depth $q$
can be constructed in quasi-linear time (Lemma~\ref{lemma:construction_time}).

Observe that each neighbor $u_k$ of $c_j$ in $\TG_{c_j}$
defines a $(c_j,u_k)$-cut in the auxiliary graph $G'_j$,
given by the two connected components of $\TG\setminus \set{(c_j,u_k)}$.
The algorithm evaluates for each $u_k$ the capacity of this cut in $G'_j$,
and verifies that it is equal to the capacity of the newly-added edge $(\set{c_j}, U_k)$ (set to be the same as of edge $(c_j, u_k)$ in $\TG$). 
In fact, all these cuts evaluations are performed not sequentially
but rather simultaneously for all $k\in [d]$, as follows. 
The key observation is that if we denote
each aforementioned $(c_j,u_k)$-cut by $(V(G'_j)\setminus C'_k, C'_k)$,
where $u_k\in C'_k$, 
then $\set{c_j},C'_1,\ldots,C'_d$ are disjoint subsets of $V(G'_j)$. 
One can clearly evaluate the capacity of all these $d$ cuts
in a single pass over the edges of $G'_j$, 
and since each edge contributes to at most two cuts (by the disjointness),
this entire pass takes only linear time $O(\card{E(G'_j)})$.

Next, to verify that each $(c_j,u_k)$-cut exhibited above,
namely, each $(V(G'_j)\setminus C'_k, C'_k)$,
is actually a minimum $(c_j,u_k)$-cut in $G'_j$, 
the algorithm finds a flow whose value is equal to the cut capacity.
In order to perform this task simultaneously for all $k\in[d]$, 
our verification algorithm employs a known result about disjoint trees,
as a witness for maximum-flow values in a graph with unit edge-capacities
(strictly speaking, this witness provides lower bounds on maximum-flow values). 
In the following theorem, a \emph{directed tree rooted at $r$} 
is a directed graph arising from an undirected tree
all of whose edges are then directed away from $r$.
This is equivalent to an arborescence
(having exactly one path from $r$ to every node other than $r$),
however we will not require that it spans all the graph nodes. In the following, $\MF_{G}(s,t)$ is the maximum $st$-flow value in a graph $G$.

\begin{lemma}\label{lemma:packing}
Given an undirected multigraph $H=(V_H,E_H)$, a root node $r\in V_H$, and a function $\lambda:V_H\rightarrow [\card{E_H}]$, it is possible to nondeterministically verify in time $\tO(\card{E_H})$ that
\begin{equation}\label{eq}
  \forall v\in V_H\setminus\set{r},
  \qquad
  \MF_{H}(r,v)\geq \lambda(v).
\end{equation}
Here, nondeterministic verification means that
if~\eqref{eq} holds then there exists a guess that leads to output ``yes'';
and if~\eqref{eq} does not hold then every guess leads to output ``no''.
\end{lemma}
%\rnote{Did we define $\MF_{H}(r,v)\geq \lambda(v)$?}

\begin{proof}
We use the following theorem known from~\cite[Theorem $2.7$]{Bang-Jensen95}, 
in its variation from~\cite{Cole03} as the Tree Packing Theorem. 

\begin{theorem}\label{theorem:BHKP}
Let $H_e$ be an Eulerian directed graph, and $r_e$ be a node in $H_e$.
%let $\MF_{H_e}(r_e,v)$ be the minimum-cut value from $r_e$ to $v$ in $H_e$.
Then there exist $\max_{v\neq r_e}\{\MF_{H_e}(r_e,v)\}$ edge-disjoint directed trees rooted at $r_e$, such that each node $v\in H_e$ appears in exactly $\MF_{H_e}(r_e,v)$ trees. 
%\rnote{For consistent notation with Lemma~\ref{lemma:packing},
%  it makes sense to replace $\lambda_{H_e}(r_e,v)$ by $\MF_{H_e}(r_e,v)$. 
%}
\end{theorem}

Given the undirected multigraph $H$, first subdivide each edge into two edges with a new node in between them,
then orient each edge in both directions to obtain an Eulerian directed graph $H_e$. 
Observe that the minimum-cut values between pairs of original nodes in $H_e$ are the same as in $H$.
Now find all maximum-flow lower-bound values from $r$
in $H_e$ by guessing $\card{V_H}$ edge-disjoint trees and then counting occurrences of each node in those trees. By Theorem~\ref{theorem:BHKP}, these counts correspond to maximum-flow lower-bound values from $r$. And so if the guessed trees support the values given by $\lambda$, then answer ``yes", and otherwise answer ``no". Note that the conversion to directed Eulerian graph multiplied the amount of edges by $2$, and so the running time is still near linear.
%
%
%to find all maximum-flow lower-bound values from $r$
%in $H_e$ by guessing $\card{V_H}$ edge-disjoint trees.
%\rnote{strictly speaking, this should read like a description of an algorithm
%  that guesses disjoint trees and counts occurrences.
%  Theorem~\ref{theorem:BHKP} is then used to prove correctness. 
%}
%
\end{proof}

The verification algorithm then applies Lemma~\ref{lemma:packing}
to $G'_j$ with $c_j$ as the root, and verifies in time $\tO(\card{E(G'_j)})$
that the maximum-flow from $c_j$ to each $u_k$
is at least the capacity of the $(c_j,u_k)$-cut exhibited above 
(in turn verified to be equal to the capacity of edge $(c_j, u_k)$ in $\TG$).

\paragraph{Correctness.}

We begin by claiming that if the guessed tree $\TG$ is a correct cut-equivalent tree of $G$, then our algorithm outputs $\TG$; we discuss the complement case afterwards.
Since $\TG$ is a cut-equivalent tree, every verification step of an expansion will not fail and so the algorithm will not terminate and output $\TG$ at the end, as required.
%First, we show that every expansion step in the verification process can also be simulated by a series of Gomory-Hu steps.
%Since the verification process starts with a single super-node containing $V$ and only does expansion steps,
%we conclude that the final result is a cut-equivalent tree.

Next, we show that if $\TG$ is not a cut equivalent tree, then our algorithm will not succeed.
This is proved mainly by the claim below, that an intermediate tree attained by expansion steps can be attained also by a sequence of Gomory-Hu steps.

\begin{claim}\label{claim:GH}
If there is a sequence of Gomory-Hu steps simulating  expansions attaining ${\T}^{(j)}$, and another expansion step is being done to attain ${\T}^{(j+1)}$, then there is a sequence of Gomory-Hu steps simulating this last step too.
\end{claim}
\begin{proof}
Under our assumption there is a truncated execution of the Gomory-Hu algorithm that produces ${\T}^{(j)}$. We describe a sequence of Gomory-Hu algorithm's steps starting with ${\T}^{(j)}$ that produces ${\T}^{(j+1)}$.
Recall that $U_{1},...,U_{d}$ are $\set{c_j}$'s neighbors in $\T^{(j+1)}$ such that $\bigcup_{i=1}^d U_i\cup \set{c_j} = V(\TG_{c_j})$,
%\rnote{but $U_1$ is a set and not a node in the tree?}
and $u_{1},...,u_{d}$ are the nodes by which the capacities of the edges $(\set{c_j},U_k)$, $k\in [d]$, were chosen.

The Gomory-Hu steps are as follows, where we denote by $\T$ the intermediate tree along the execution.
Starting with $\T=\T^{(j)}$, for $k=1,...,d$, the Gomory-Hu execution picks the pair $c_j,u_{k}$ from the super-node containing it in $\T$ as the pair $s,t$ in the Gomory-Hu algorithm description (see the description in Section~\ref{GHU_Alg}), and the given minimum-cut value between them is asserted. Then, for the partitioning of this super-node in $\T$, the execution picks the minimum-cut between $c_j,u_{k}$ as in $\TG$ (which is a minimum cut also in the corresponding auxiliary graph) and modifies the intermediate tree accordingly. Note that the last expansion step was assumed to be successful (i.e., verified correctly), thus all the cuts chosen for the partitioning are minimum-cuts.
%
%
%To see why the Gomory-Hu algorithm can continue choosing those cuts, note that for every cut between $c_j$ and $u_{k}$ that is chosen, the side of $c_j$ contains all the nodes $u_{k+1},...,u_{d}$.
%
%\rnote{This paragraph is not clear.
%  Perhaps start by explaining the problem you are trying to address
%  (i.e., what is missing in the previous paragraph). 
%}
\end{proof}

Now, assume for the contrary that $\TG$ is not a cut-equivalent tree of $G$ and our algorithm still produces it.
As a consequence of Claim~\ref{claim:GH}, there is a sequence of Gomory-Hu steps attaining $\TG$, contradicting the proof of correctness of the Gomory-Hu algorithm (which cannot produce $\TG$).
Thus, it is impossible that our algorithm finishes and produces $\TG$, and so in one of the minimum-cut verifications after an expansion step, the cut witness inspired from $\TG$ would not be correct, or there would not be a set of directed trees to testify that the corresponding cuts are minimal. This completes the proof of correctness.
%\end{proof} % [Proof of Theorem~\ref{nondet:main}]

\paragraph{Running Time.}
Observe that the running time of a single expansion step,
i.e., verifying its corresponding minimum cuts by evaluating cuts and flows,
is quasi-linear in the size of the auxiliary graph.
Thus, we only have to show that the total size of all the auxiliary graphs
(over all the expansions) is quasi-linear. 
The next lemma provides a bound for a single depth $q$.
As a corollary and since the depth of the entire decomposition is $O(\log n)$,
we get a bound of $\tO(m)$ on the total size of all auxiliary graphs over all depths.

\begin{lemma}\label{lemma:size}
Let $D_q=\{c_{j_1},\dots,c_{j_2}\}$ contain the centroids at depth $q$. 
Then the total size of $G'_{j_1},\dots,G'_{j_2}$ is at most $O(m)$.
\end{lemma}
\begin{corollary}\label{cor:total}
The total size of all auxiliary graphs
(over all depths) is $\tO(m)$.
\end{corollary}
\begin{proof}[proof of Lemma~\ref{lemma:size}]
%First, consider edge appearances $uv$ where either $u$ or $v$ are in $V({\T_{c_j}})$ in the auxiliary graph $G'_j$.
%Clearly, these appearances for every edge are at most two at this depth $q$, i.e., at most in both $G'_j$ and $G'_{j'}$ where $u\in V({\T_{c_j}})$ and $v\in V({\T_{c_{j'}}})$.
%
%Thus their total number is at most $O(m)$. 
%
%Second, consider edge appearances $uv$ such that neither $u$ nor $v$ are in $V({\T_{c_j}})$ for $G'_j$.
%

We count for each edge $uv\in E(G)$ in how many auxiliary graphs it appears in depth $q$. This turns out to be at most $2+(\dist_{\T}(u,v)-1)$ where $\dist_{\T}(u,v)$ is the hop-distance, i.e., the minimum number of edges (ignoring weights or capacities) in a path 
between $u$ and $v$ in the tree $\T$.
The $2$ term is from edges $uv$ such that either $u$ or $v$ are in $V({\T_{c_j}})$ in the auxiliary graph $G'_j$. Clearly, every such edge is in at most two auxiliary graphs at this depth $q$, i.e., at most in both $G'_j$ and $G'_{j'}$ where $u\in V({\T_{c_j}})$ and $v\in V({\T_{c_{j'}}})$.
The $(\dist_{\T}(u,v)-1)$ term is a bound on appearances of edges $uv$ such that neither $u$ nor $v$ are in $V({\T_{c_j}})$ for $G'_j$, which is proved in the claim below.
While our graph has unit capacities, the next claim is for general capacities. In what follows, $c_{\T}(e)$ the capacity of the edge $e$ in $\T$.

\begin{claim}\label{claim:technical}
For every cut-equivalent tree $\T$ of a graph $G$ with edge capacities $c_G:E\rightarrow \mathbb{R}_+$,
$$
\sum_{uv\in E(G)} c_G(u,v)\cdot\dist_{\T}(u,v)\leq 2\sum_{uv\in E(G)} c_G(u,v).
$$
\end{claim}
\begin{proof}
  We first show that $\sum_{uv\in E_{G}} c_G(u,v)\cdot\dist_{\T}(u,v) = \sum_{e\in E_{\T}} c_{\T}(e)$,
  where $c_{\T}$ denotes edge capacity in $\T$.
  Observe that each $c_{\T}(e)$ is the value of a certain cut in $G$,
  hence we can evaluate the right-hand side differently,
  by summing over the graph edges $uv\in E(G)$
  and counting for each edge in how many such cuts it appears.
  Recalling that $\T$ is a cut-equivalent tree,
  the count for each graph edge $uv\in E(G)$ is exactly $\dist_{\T}(u,v)$ contributions of $c_G(u,v)$,   giving altogether the left-hand side.

  Second, we show $\sum_{uv\in E(\T)} c_{\T}(u,v)\leq 2\sum_{uv\in E(G)} c_G(u,v)$.
  To see this, observe that $c_{\T}(u,v)\leq \min\{\deg_{c_G}(u), \deg_{c_G}(v)\}$
  where $\deg_{c_G}(u)$ is the total capacity of edges incident to $u$.
  %since $G$ has unit-capacity edges.
  Now fix a root vertex in $\T$,
  and bound each tree edge by $c_{\T}(u,v)\leq \deg_{c_G}(v)$,
  where $v$ is the child of $u$ (i.e., farther from the root) in $\T$. 
  Summing this bound over all the tree edges, observing that the corresponding vertices $v$ are all distinct,
  and the proof follows.
\end{proof}

Recall that by Observation~\ref{observation:disjoint} the super-nodes $V(\T_{c_{j_1}}),\ldots,V(\T_{c_{j_2}})$ of the same depth $q$ are pairwise disjoint. Thus, an edge $uv$ appears in at most $dist_{\TG}(u,v)-1$ auxiliary graphs of depth $q$, 
which totals to $O(m)$ for all the edges in this depth according to the unit edge-capacity special case of the above Claim~\ref{claim:technical}. This concludes Lemma~\ref{lemma:size}.

\end{proof}

Next, we bound the time it takes to construct all the auxiliary graphs.% for a certain depth $q$.
\begin{lemma}\label{lemma:construction_time}
The total time it takes to construct the auxiliary graphs for all the expansions in the centroid decomposition is $\tO(m)$.
\end{lemma}
\begin{proof}
Let $c_j$ be a node that is expanded at some depth $q\geq 1$, and let $c_{j,1},\dots,c_{j,d}$ be the expanded nodes in $U_{1},\dots,U_{d}$, respectively at depth $q+1$ (or just $\bot$ for singletons).
%\rnote{double indices could be confusing (unless expected, like for a matrix),
%  so please use a comma, like $c_{j,1}$.
%}
Note that $G'_{j,1},\dots,G'_{j,d}$ (whichever exist) can all be constructed in total time that is linear in the size of $G'_j$.
Thus, the total time it takes to construct the auxiliary graphs for all the expansions at a single depth $q$ is linear in the size of the auxiliary graphs in the parent depth. Since the construction of the auxiliary graph of depth $0$ (i.e., the entire graph) can trivially be done in time $O(m)$ time, it follows by corollary~\ref{cor:total} that the construction time of the auxiliary graphs for all the expansions takes at most $\tO(m)$ time.
%\rnote{the last sentence is confusing me; why talk about expansions?
%  and if you meant auxiliary graphs, then doesn't this sentence repeat
%  the earlier sentence ``Thus the total time...''?
%}
\end{proof}

\subsection{Reduction from a Decision Problem to a Total Function}\label{subsec:reduction}

Let us start with the formal statement of NSETH.
\begin{hypothesis}[Nondeterministic Strong Exponential-Time Hypothesis (NSETH)]
For every $\varepsilon>0$ there exists $k=k(\varepsilon)$ such that $k$-TAUT (the language of all $k$-DNF formulas that are tautologies) is not in $\NTIME(2^{n(1-\varepsilon)})$.
\end{hypothesis}
Note that deciding if a $k$-DNF formula is a tautology is equivalent to deciding if a $k$-CNF formula is satisfiable, thus the above hypothesis could be stated also using $k$-CNF appropriately.
Next, we define (deterministic) fine-grained reductions from a decision problem to a total function. Note that these are Turing reductions.

\begin{definition}[Fine-Grained Reduction from a Decision Problem to a Total Function]
Let $L$ be a language and $P$ be a total function, and let $T_L(\cdot)$ and $T_P(\cdot)$ be time bounds. We say that $(L,T_L)$ admits a fine-grained reduction to $(P,T_P)$ if for all $\eps>0$ there is a $\gamma>0$ and a deterministic Turing machine $M^P$ (with an access to an oracle that generates a solution to every instance of $P$) such that:

\begin{enumerate} 
\renewcommand{\theenumi}{\alph{enumi}}
\item \label{itit:a}
	$M^P$ decides $L$ correctly on all inputs when given a correct oracle for $P$. 
\item \label{itit:b}
  Let $\tilde{Q}(M^P,x)$ denote the set of oracle queries made by $M^P$ on input $x$ of length $n$. Then the query lengths obey the bound
$$
\forall x, \qquad
\DTIME(M^P,\card{x}) + \sum_{q\in \tilde{Q}(M,x)}(T_P(\card{q}))^{1-\varepsilon}\leq (T_L(n))^{1-\gamma}.
$$ 
%\item \label{itit:c}
%	The size of the answers given by the $tf$ oracle is bounded. \anote{did we say this is redundant?}
%$$
%\sum_{q\in \tilde{Q}(M,x)}(\card{tf(q)})\leq (T_1(n))^{1-\delta}
%$$
\end{enumerate}
\end{definition}

We are now ready to prove the non-reducibility result under NSETH for total functions with small nondeterministic complexity. The proof arguments are similar to those of Carmosino et al. \cite{carmosino2016nondeterministic}.

\begin{theorem}\label{nondet:if}
Suppose $P$ is a total function with nondeterministic time complexity $T(m)$.
If for some $\delta>0$ there is a deterministic fine-grained reduction from $k$-SAT with time-bound $2^n$ to $P$ with time bound $T(m)^{1+\delta}$, i.e. from $(k\text{-SAT},2^n)$ to $(P,T(m)^{1+\delta})$, then NSETH is false. 
\end{theorem}

\begin{proof}
We will use the assumption of the theorem to describe a nondeterministic algorithm for $k$-TAUT that refutes NSETH.
Let $\phi$ be an instance of $k$-TAUT, and note that $\phi \in k$-TAUT iff $\neg \phi \notin k$-SAT.
Our nondeterministic algorithm $A$ first computes the CNF formula $\neg \phi$, then simulates the assumed reduction $M_1$ from $k$-SAT to $P$ on $\neg \phi$, and eventually outputs the negation of the simulation's answer, or $Reject$ if the simulation returns $\bot$. 
%Recall that our nondeterministic algorithm $A$ for $k$-TAUT must always output $Reject$ if the input is not a $k$-DNF formula, and it needs to have at least one guess on which it outputs $Accept$ if the input is a tautology.

Let $M_2$ be the Turing Machine showing that $P$ has nondeterministic time complexity $T(m)$. 
Whenever the reduction $M_1$ produces a query to $P$, our algorithm $A$ executes $M_2$ on this query with some guess string $g$. Let $g_i$ be the guess string used for the $i^{th}$ query to $P$ made by $M_1$. If any of the executions of $M_2$ throughout the simulation outputs $\bot$, then $A$ stops and outputs $Reject$. Otherwise (all executions output valid answers), the simulation continues until $M_1$ terminates. At this point, the output of $M_1$ must be correct, and our algorithm $A$ outputs the opposite answer.

Let us argue about the correctness of our algorithm. First, it only outputs $Accept$ if the guesses and all answers to the $P$-queries were correct and then $M_1$ rejected, meaning that $\neg \phi \notin k$-SAT i.e., $\phi \in k$-TAUT.
Second, for every yes-instance $\phi \in k$-TAUT there is at least one sequence of guesses $g_1,g_2,\ldots$ that makes $A$ output $Accept$, due to the correctness of the reduction $M_1$ and the fact that $M_2$ nondeterministically computes $P$ correctly.
%@@@@@@@
Finally, the running time of $A$ can be upper bounded by
$$
\DTIME(M_1) + \sum_{q\in \tilde{Q}(M_1,x)}T(\card{q}) \leq \DTIME(M_1)+\sum_{q\in \tilde{Q}(M_1,x)}T(\card{q})^{1+\delta-\eps}\leq (2^n)^{1-\varepsilon'}
$$
for $0<\varepsilon<\delta$ where the last inequality is due to the reduction from $k$-SAT to $P$, $\DTIME(M_1)$ is the time of operations done by $M_1$, $\tilde{Q}(M_1,x)$ is the queries made by $M_1$ to the $P$-oracle on an input $x$, and the last inequality follows for some $\varepsilon'(\varepsilon)>0$ because $M_1$ is a correct fine-grained reduction.
Thus, $A$ refutes NSETH.

%A deterministic Turing reduction from SAT to $P$ would falsify NSETH;
%any such reduction would create instances $X_1,...,X_k$ of $P$ in an adaptive way, i.e., that the input $X_i$ and any solution for it $G_i\in \SOL(G_i)$ may affect the instances $X_j$ for $j>i$, and the answer for the given UNSAT formula is a function of $X_1,...,X_k$. Thus, there is a series of guesses (related to item $a$) that will make our algorithm output all those solutions and finally decide the UNSAT answer correctly, and any series of guesses could end up with solutions or "don't know" (item $b$), as required by a nondeterministic algorithm.
%Regarding running time, if we apply the reduction and solve each $X_i$ in time $T(\card{E(X_i)})$ rather than at least $T(\card{E(X_i)})^{1+\delta}$, the total running time would clearly be significantly faster than $O(2^n)$, as required.
%%Regarding running time, in the deterministic fine-grained reduction it must be that the total running time for constructing cut-equivalent trees for every $X_i$'s is at least $T(\card{X_i})^{1+\delta-o(1)}$ and totally $\sum_{i\in [k]} T(\card{X_i})^{1+\delta-o(1)}=T(\card{X_i})^{1+\delta-o(1)}$ (otherwise SETH would fail)
%%
%%$\sum_{i\in [k]} T(\card{X_i)$
%%\rnote{I don't see $\delta$ in the proof}
\end{proof}

Since the construction of a cut-equivalent tree is a total function, and by theorem~\ref{Theorem:Algorithm} its nondeterministic complexity is $\tO(m)$, applying Theorem~\ref{nondet:if} implies that any deterministic reduction from SETH to the construction of a cut-equivalent tree that implies a lower bound of $\Omega(m^{1+\delta})$, for some $\delta>0$, would refute NSETH, concluding Theorem~\ref{thm3}.
%\texttt{\rnote{``with time bound $\Omega(m^{1+\delta})$'' sounds wrong;
%  it seems to describe the words reduction or construction. 
%}
%  \onote{a similar source of confusion might be also in Theorem~\ref{nondet:if}, perhaps we should change the wording there as well?}
%  \anote{i changed the theorem a bit. is it better now?}

\section{Conditional Lower Bound for  \APMF}\label{sec:lower}

In this section we prove a conditional lower bound for \APMF
in undirected graphs with node capacities.
Our construction is inspired by the one in~\cite{KT18}, which was designed for directed graphs with edge capacities, but it adopts it using our new trick described in the introduction.
In fact, readers familiar with the reduction in~\cite{KT18} may notice that we had to tweak it a little, making it simpler in certain ways but more complicated in others.  This was necessary in order to apply our new trick successfully to it.

The starting point for our reduction is the \TOV problem.
\begin{definition}[\TOV]
Given three sets $U_1, U_2, U_3 \subseteq \{0,1\}^d$ containing $n$ binary vectors each, over dimension $d$, decide if there is a triple $(\alpha,\beta,\gamma)$ of vectors in $U_1 \times U_2 \times U_3$, whose dot product is $0$. That is, a triple for which for all $i \in [d]$ at least one of $\alpha[i],\beta[i],\gamma[i]$ is equal to $0$. 
\end{definition}

An adaptation of the reduction by Williams \cite{Wil05} shows that  \TOV cannot be solved in $O(n^{3-\eps})$ time for any $\eps>0$ and $d=\omega(\log{n})$, unless SETH is false (see \cite{ABV15}).
For us, it suffices to assume the milder conjecture that \TOV cannot be solved in  $O(n^{3-\eps})$ time when $d=n^{\delta}$, for all $\eps,\delta>0$.
Refuting this conjecture has important implications beyond refuting SETH \cite{GIKW17,ABDN18}, e.g. it refutes the Weighted Clique Conjecture.

The high level structure of the reduction is the following: create three layers of nodes that correspond to the three sets of vectors, with additional two layers in between them that correspond to the coordinates. These additional layers help keep the number of edges small by avoiding direct edges between pairs of vectors.
Among other things, we utilize the trick described in the introduction and set the capacity of the nodes in the leftmost and rightmost sides to be $1$, while making the other capacities much larger. This way a flow would not gain too much from crisscrossing through these nodes. Formally, we prove the following.

\begin{lemma}\label{ProofsCapLemma}
\TOV over vector sets of size $n$ and dimension $d$ can be reduced to \APMF in undirected graphs with $\Theta(n\cdot d)$ nodes, $\Theta(n\cdot d)$ edges, and node capacities in $[2n^2d]$.
\end{lemma}

\begin{proof}
Given a \TOV instance $F$ we construct a graph $G$ with maximum flow size between some pair (among a certain set of pairs) bounded by a certain amount if and only if $F$ is a yes instance. 
For simplicity, we first provide a construction that has some of the edges directed (only where we will specifically mention that), and then we show how to avoid these directions. In addition, some of the edges will be capacitated as well, however the amount of such edges is small enough so that subdividing them with appropriate capacitated nodes will work too without a significant change to the size of the constructed graph. 

\paragraph{An Intermediate Construction with Few Directed Edges.}
To simplify the exposition, we start with a construction of a graph $G'$ in which most of the edges are undirected, but some are still directed (see figure~\ref{Figs:REDUCTION_CAP}).

\begin{figure}[!ht]
%   \centering
       \includegraphics[width=0.9\textwidth,left]{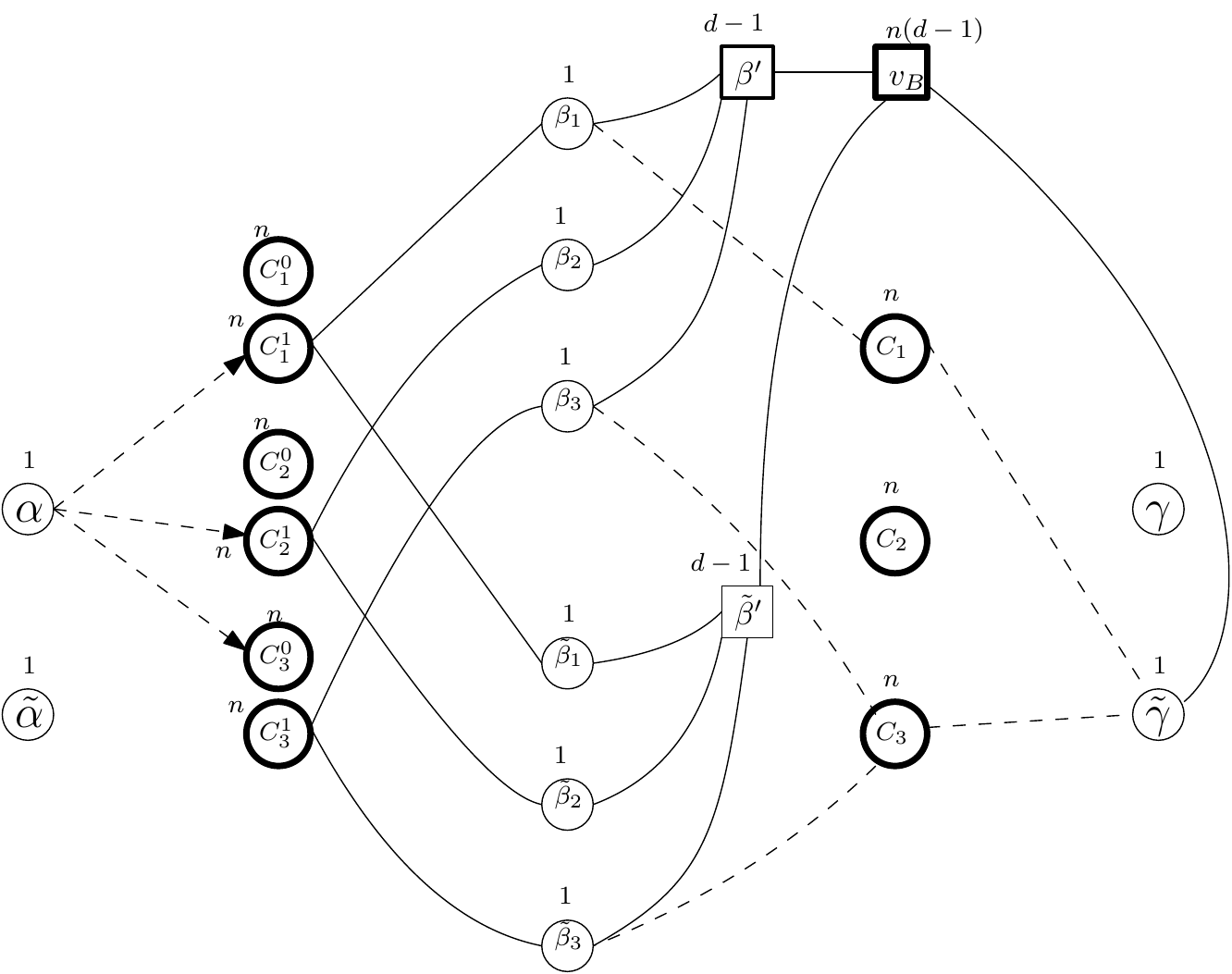}
   \caption[-]{
   An illustration of part of the reduction. Here, $U_1$, $U_2$, and $U_3$ have two vectors each; $\alpha$ and $\tilde{\alpha}$ in $U_1$, $\beta$ and $\tilde{\beta}$ in $U_2$, $\gamma$ and $\tilde{\gamma}$ in $U_3$. Bolder nodes correspond to nodes of higher capacity, and dashed edges are conditional on the input instance. For simplicity, we omit the edges not relevant to $\alpha$ and $\tilde{\gamma}$, and also the edges from nodes in $\{C_i^0\}_{i\in[3]}$ to nodes in $\{\beta',\tilde{\beta}'\}$. In this illustration, $\alpha=110$, $\beta=101$, $\tilde{\beta}=001$, and $\tilde{\gamma}=101$. Note that the triple $\alpha$, $\tilde{\beta}$, and $\tilde{\gamma}$ has an inner product $0$, and indeed the maximum flow from $\alpha$ to $\tilde{\gamma}$ is $2\cdot 3 -1=5$.
   }
   \label{Figs:REDUCTION_CAP}
\vspace{.1in}\hrule
\end{figure}

Our final graph $G$ will be very similar to $G'$. It will have the same nodes and edges except that all edges will be undirected and the capacities on the nodes will be a little different. 
%%%Amir: note that i changed from $G$ to $G'$.

We construct the graph $G'$ on $N$ nodes $V_1\cup V_2\cup V_3\cup A\cup B\cup \{v_B\}$. 
The layer $V_1$ contains a node $\alpha$ of capacity $1$ for every vector $\alpha\in U_1$.
$V_2$ contains $d+1$ nodes for every vector $\beta\in U_2$, $d$ nodes denoted by $\beta_i$ for every $i\in[d]$ and their capacity is $1$, plus a node denoted by $\beta'$ of capacity $d-1$. 
$V_3$ contains a node $\gamma$ of capacity $1$ for every vector $\gamma$ in $U_3$.
The intermediate layer $A$ contains $2d$ nodes: two nodes $C_i^{0}$ and $C_i^{1}$ of capacity $n$ for every coordinate $i\in [d]$. 
The other intermediate layer $B$ contains a node $C_i$ of capacity $n$ for every coordinate $i\in [d]$. 
Finally, the auxiliary node $v_B$ has capacity $n(d-1)$. 
With a slight abuse of notation, we will use the following symbols in the following ways: $\alpha$ will be either a node in $V_1$ or a vector in $U_1$; $\beta$ will be a vector in $U_2$;  $\gamma$ will be either a node in $V_3$ or a vector in $U_3$; and $C_i$ will be either a node in $B$ or a coordinate in $[d]$. The usage will be clear from context.

The edges of the network will be defined as follows. First, we describe the edges that depend on the given \TOV instance.
%In order to simplify the reduction, we partition the edges into red and blue colors, as follows.
\begin{itemize}
\item For every $\alpha$ and $i\in [d]$, we add a \emph{directed} edge from $\alpha$ to $C_i^{0}$ if $\alpha [i]=0$, and a directed edge from $\alpha$ to $C_i^{1}$ if $\alpha [i]=1$. 
\item For every $\beta$, we add an (undirected) edge from $\beta_i$ to $C_i$ if $\beta [i]=1$.
\item  For every $\gamma$ and $i\in [d]$, we add an (undirected) edge from $C_i$ to $\gamma$ if $\gamma [i]=1$.
\end{itemize}
Moreover, there will be some (undirected) edges that are independent of the vectors.
For every $\beta$, we have an edge of capacity $1$ from $C_i^{0}$ to ${\beta}'$, and an edge of capacity $1$ from $C_i^{1}$ to $\beta_i$.
Also, for every $\beta$, we have an edge from $\beta_i$ to $\beta'$, and an edge from ${\beta}'$ to $v_B$. 
%For every $\beta$ we add an edge from $\beta_i$ to $C_i$ if $\beta [i]=1$, 
Finally, for every $\gamma$, we have an edge from $v_B$ to $\gamma\in V_3$. (Unless specified otherwise, these edges have no capacity constraints.)

The graph built has $N=n+2d+n\cdot d+n+1+d+n=\Theta(nd)$ nodes, at most $O(nd)$ edges, all of its capacities are in $[N]$, and its construction time is $O(Nd)$.
% (see Figure~\ref{Figs:REDUCTION_CAP}). 
%

\medskip
The following two claims prove the correctness of this intermediate reduction.

\begin{claim}\label{claim:simple}
If every triple of vectors in $(U_1, U_2, U_3)$ have inner product at least $1$, then for all pairs $\alpha \in V_1,\gamma \in V_3$ the maximum-flow in $G'$ is at least $n\cdot d$.
\end{claim} 
\begin{proof}
Assume that every triple of vectors in $(U_1, U_2, U_3)$ has inner product at least $1$, and fix some $\alpha$ and $\gamma$. 
We will explain how to send $n\cdot d$ units of flow from $\alpha$ to $\gamma$ in $G'$.
By the assumption, for every $\beta$ there exist an $i\in [d]$ such that $\alpha[i]=\beta[i]=\gamma[i]=1$, and denote this index by $i_{\beta}$. 
Each $i_{\beta}$ induces a path $(\alpha\rightarrow C_{i_{\beta}}^1\rightarrow\beta_{i_{\beta}}\rightarrow C_{i_{\beta}}\rightarrow \gamma)$ from $\alpha$ to $\gamma$, and so we pass a single unit of flow through every one of them, in what we call the first phase. Note that so far, the flow sums up to $n$, and we carry on with describing the second phase of flow through nodes of the form $\beta'$.

We claim that for every $\beta$, an additional amount of $(d-1)$ units can pass through $\beta'$, which would add up to a total flow of $n(d-1)+n=nd$, concluding the proof. Indeed, for every $\beta$, we send flow in the following way. For every $i\in [d]\setminus i_{\beta}$, if $\alpha[i]=1$ then we send a single unit through $(\alpha\rightarrow C_i^1\rightarrow \beta_i\rightarrow \beta' \rightarrow  v_B\rightarrow \gamma)$, and otherwise we send a unit of flow through $(\alpha\rightarrow C_i^0\rightarrow \beta'\rightarrow v_B\rightarrow \gamma)$.

Since we defined the flow in paths, we only need to show that the capacity constraints are satisfied. Nodes of the form $C_i$ are only used in the first phase, and the flow through them equals $n$ in total, and so their flow is within the capacity. The node $v_B$ is only used in the second phase and has $n(d-1)$ units of flow passing through it, just as its capacity. 
For every $\beta$ and $i=i_{\beta}$, we pass in the first phase a single unit of flow through $\beta_i$. For every $\beta$ and $i\neq i_{\beta}$, we transfer in the second phase a unit of flow through $\beta_i$ if and only if $\alpha[i]=1$, thus it is bounded. For every $\beta'$, we pass in the second phase exactly $(d-1)$ units of flow through $\beta'$, preserving its capacity. For every $C_i^j\in N(\alpha)$ with $i\in [d]$ and $j\in \{0,1\}$, we pass a total of $n$ units of flow to nodes in $V_2$, one unit on each edge, thus the capacities are preserved, concluding the proof.

\end{proof}

\begin{claim}
If there is a triple of vectors $(\alpha_{\Phi}, \beta_{\Phi}, \gamma_{\Phi})\in (U_1, U_2, U_3)$ whose inner product is $0$, then the maximum-flow in $G'$ from $\alpha_{\Phi} \in V_1$ to $\gamma_{\Phi} \in V_3$ is at most $nd-1$.
\end{claim} 
\begin{proof}
Assume for contradiction that there exists such a flow of value at least $nd$, and denote it by $f$. 
Let $f=\{p_1,...,p_{\card{f}}\}$ be a description of $f$ as a (multi-)set of paths of single units of flow. 

For a node $x$, denote by $N(S)$ the set of all nodes adjacent to $x$.
By our construction, the total capacity of all nodes in $N(\alpha_{\Phi})$ sums up to $nd$ exactly. 
Therefore, $f$ must have all of the nodes in $N(\alpha_{\Phi})$ saturated.

%Since the total capacity of the nodes adjacent to $\alpha_{\Phi}$, denoted $N(\alpha_{\Phi})$, sums up to $nd$, $f$ must have all of the nodes in $N(\alpha_{\Phi})$ saturated. 
Consider a node $C_i^j\in N(\alpha_{\Phi})$ for some $i\in [d]$ and $j\in \{0,1\}$.
Note that $C_i^j$ is saturated in $f$ while its capacity is $n$ and it has exactly $n$ edges adjacent to it (excluding the edges incoming from $V_1$) of capacity $1$ each.
Therefore, we get that every node in $N(C_i^j)\setminus V_1$ must receive a single unit of flow from $C_i^j$ in $f$. 
Hence, every $\beta$-cloud, which we define as all the nodes that are associated with a $\beta$, must have exactly $d$ flow paths in $f$ for which it is the first $\beta$-cloud that they pass through. We call this a \textit{first passing} of a path through a $\beta$-cloud.
In particular, for every $\beta$ and for every $i\in [d]$ such that $\alpha_{\Phi}[i]=1$ there must be a path $p_{\beta,i}$ in $f$ whose prefix is $(\alpha_{\Phi}, C_i^1, \beta_i,...)$. 

Our main claim is that the $\beta_{\Phi}$-cloud can only have up to $d-1$ flow paths that are first passing through it. 
Clearly, if there are more, then at least one of them does not pass through ${\beta_{\Phi}}'$ (whose capacity is only $d-1$), so name this path $p'$. 
We will argue that this path must be in conflict with one of the $p_{\beta,i}$ paths described above.

For some $i\in [d]$ the prefix of $p'$ must be $(\alpha_{\Phi}, C_i^1, \beta_{\Phi_i}, C_i,...)$, since this is the only way it can avoid the node ${\beta_{\Phi}}'$.
This can only happen for an $i \in [d]$ for which $\alpha[i]=\beta[i]=1$, or else those edges will not exist in $G$.
But since $(\alpha_{\Phi}, \beta_{\Phi}, \gamma_{\Phi})$ is a triple whose inner product is $0$, it must be that $\gamma_{\Phi} [i]=0$ and so the edge $\{C_i,\gamma\}$ is not in the graph. 
Hence, after $C_i$ this path can only go to a node $\tilde{\beta}_i$ for some $\tilde{\beta}$, and the (longer) prefix of $p'$ must be $(\alpha_{\Phi}, C_i^1, \beta_i, C_i, {\tilde{\beta}}_i ,...)$. 
Note that this is the same index $i$, and we know that $\alpha_{\Phi} [i] =1$.
Therefore, by the above, we know that there is another path $p_{\tilde{\beta},i}$ in $f$ that has ${\tilde{\beta}}_i$ as the third node on the path. 
(That is, there is already a path that is first-passing through ${\tilde{\beta}}_i$.) 
This is a contradiction to the fact that the capacity of ${\tilde{\beta}}_i$ is $1$.

%%%Amir: since this is the most difficult claim, it may be good to expand the proof a bit.
\end{proof}

\paragraph{The Final Construction.}

The main issue with avoiding the directions on the edges between nodes in $V_1$ and $A$, is that additional $\alpha$'s might participate in the flow as well, potentially allowing one additional unit of flow to pass through. As described in the introduction, the solution is to multiply the capacities of all nodes that are not in $V_1\cup V_3$ by $2n$. This is how we get our final graph $G$ from $G'$. 
In the following we show how this modification concludes the proof of Lemma~\ref{ProofsCapLemma}.

\begin{claim}\label{claim:first}
If every triple of vectors in $(U_1, U_2, U_3)$ has inner product at least $1$, then for all pairs $\alpha \in V_1,\gamma \in V_3$ the maximum-flow in $G$ is at least $2n^2d$.
\end{claim} 
\begin{proof}
Since the flow that was defined in Claim~\ref{claim:simple} does not touch nodes in $V_1\cup V_3$, considering the same flow in $G$ but multiplied by $2n$, we get a new flow that is of size $nd\cdot(2n)$, concluding the proof.
\end{proof}

\begin{claim}\label{claim:second}
If there is a triple of vectors $(\alpha_{\Phi}, \beta_{\Phi}, \gamma_{\Phi})\in (U_1, U_2, U_3)$ whose inner product is $0$, then the maximum-flow in $G$ from $\alpha_{\Phi} \in V_1$ to $\gamma_{\Phi} \in V_3$ is at most $2n^2d -1$.
\end{claim} 
\begin{proof}
Let $f$ be the maximum flow from $\alpha_{\Phi}$ to $\gamma_{\Phi}$ in $G$.
The paths in $f$ can be divided into two kinds: paths that pass through nodes in $(V_1\cup V_3)\setminus \{\alpha_{\Phi},\gamma_{\Phi}\}$, and paths that do not.
The total contribution of paths of the first kind can be upper bounded by the size of $(V_1\cup V_3)\setminus \{\alpha_{\Phi},\gamma_{\Phi}\}$, which is $2n-2$, since the capacity of all nodes in this set is $1$.
On the other hand, paths from the second kind must obey the directions of the directed edges in $G'$ and can therefore be used in $G'$, except that in $G$ their multiplicity (the amount of flow we push through them) can be larger by a factor of $2n$.
Therefore, we can upper bound the total contribution of paths of the second kind by $2n$ times the maximum flow in $G'$, which is $(nd-1)(2n)$.
Thus, the overall flow is at most $(nd-1)(2n)+2n-2=2n^2d-2$, which proves Claim~\ref{claim:second}.
\end{proof}
Since we showed a gap of at least one unit of flow between the yes and the no instances, the proof of Lemma~\ref{ProofsCapLemma} is concluded.
\end{proof}

\section{Open Problems}\label{Section:Conclusion}

Many gaps and open questions around the complexity of maximum flow remain after this work. We highlight a few for which our intuitions may have changed following our discoveries.

\begin{itemize}
\item Can we break the $mn$ barrier also when the graphs have arbitrary (polynomial) capacities? Our result gives hope that this may be possible.

\item Can we reduce the directed case to the undirected, node-capacitated case? Because of our lower bound, it is likely that both of these cases will end up having the same time complexity, and so such a reduction may be possible.

\item Can we generalize the nondeterministic algorithm to be for arbitrary edge-capacities? Note that the only barrier for achieving that goal is finding lower bounds witness for flows from a certain source to other nodes.

\item Can we prove any conditional lower bound for All Pairs Max Flow in undirected graphs with edge capacities? This is obviously the most important and intriguing open question in this context. Our new deterministic and nondeterministic upper bounds make this task more challenging than previously thought. 
\end{itemize}

\section{Acknowledgements}
We would like to thank Arturs Backurs for asking about the nondeterministic complexity of the problems, Marvin Kunnemann for pointing out the connection to certifying algorithms, and Richard Peng for helpful comments on the different known upper bounds for \MF.

{\small
\bibliographystyle{alphaurlinit}
\bibliography{robi}

\newcommand{\etalchar}[1]{$^{#1}$}
\begin{thebibliography}{BENW16}

\bibitem[ABDN18]{ABDN18}
A.~Abboud, K.~Bringmann, H.~Dell, and J.~Nederlof.
\newblock More consequences of falsifying {SETH} and the orthogonal vectors
  conjecture.
\newblock In {\em Proceedings of the 50th Annual {ACM} {SIGACT} Symposium on
  Theory of Computing, {STOC} 2018}, pages 253--266, 2018.
\newblock \href {http://dx.doi.org/10.1145/3188745.3188938}
  {\path{doi:10.1145/3188745.3188938}}.

\bibitem[ABMR11]{alkassar2011verification}
E.~Alkassar, S.~B{\"o}hme, K.~Mehlhorn, and C.~Rizkallah.
\newblock Verification of certifying computations.
\newblock In {\em International Conference on Computer Aided Verification},
  pages 67--82. Springer, 2011.

\bibitem[ABW15]{ABV15}
A.~Abboud, A.~Backurs, and V.~V. Williams.
\newblock Tight hardness results for {LCS} and other sequence similarity
  measures.
\newblock In {\em {IEEE} 56th Annual Symposium on Foundations of Computer
  Science, {FOCS} 2015}, pages 59--78, 2015.
\newblock \href {http://dx.doi.org/10.1109/FOCS.2015.14}
  {\path{doi:10.1109/FOCS.2015.14}}.

\bibitem[ACZ98]{ACZ98}
S.~R. Arikati, S.~Chaudhuri, and C.~D. Zaroliagis.
\newblock All-pairs min-cut in sparse networks.
\newblock {\em J. Algorithms}, 29(1):82--110, 1998.

\bibitem[AGI{\etalchar{+}}19]{A+18}
A.~Abboud, L.~Georgiadis, G.~F. Italiano, R.~Krauthgamer, N.~Parotsidis,
  O.~Trabelsi, P.~Uznanski, and D.~Wolleb-Graf.
\newblock {Faster Algorithms for All-Pairs Bounded Min-Cuts}.
\newblock In {\em 46th International Colloquium on Automata, Languages, and
  Programming (ICALP 2019)}, volume 132, pages 7:1--7:15, 2019.
\newblock \href {http://dx.doi.org/10.4230/LIPIcs.ICALP.2019.7}
  {\path{doi:10.4230/LIPIcs.ICALP.2019.7}}.

\bibitem[AV18]{AnariV18}
N.~Anari and V.~V. Vazirani.
\newblock Planar graph perfect matching is in {NC}.
\newblock In {\em 59th IEEE Annual Symposium on Foundations of Computer
  Science}, FOCS '18, pages 650--661. {IEEE} Computer Society, 2018.
\newblock \href {http://dx.doi.org/10.1109/FOCS.2018.00068}
  {\path{doi:10.1109/FOCS.2018.00068}}.

\bibitem[AVY15]{AVY15}
A.~Abboud, V.~{Vassilevska Williams}, and H.~Yu.
\newblock Matching triangles and basing hardness on an extremely popular
  conjecture.
\newblock In {\em Proc.\ of 47th STOC}, pages 41--50, 2015.

\bibitem[BENW16]{BENW16}
G.~Borradaile, D.~Eppstein, A.~Nayyeri, and C.~Wulff{-}Nilsen.
\newblock All-pairs minimum cuts in near-linear time for surface-embedded
  graphs.
\newblock In {\em 32nd International Symposium on Computational Geometry, SoCG
  2016}, pages 22:1--22:16, 2016.

\bibitem[BFJ95]{Bang-Jensen95}
J.~{Bang-Jensen}, A.~Frank, and B.~Jackson.
\newblock Preserving and increasing local edge-connectivity in mixed graphs.
\newblock {\em SIAM J. Discret. Math.}, 8(2):155--178, 1995.
\newblock \href {http://dx.doi.org/10.1137/S0036142993226983}
  {\path{doi:10.1137/S0036142993226983}}.

\bibitem[BHKP07]{BHKP07}
A.~Bhalgat, R.~Hariharan, T.~Kavitha, and D.~Panigrahi.
\newblock An {$O(mn)$} {G}omory-{H}u tree construction algorithm for unweighted
  graphs.
\newblock In {\em 39th Annual ACM Symposium on Theory of Computing}, STOC'07,
  pages 605--614. ACM, 2007.
\newblock \href {http://dx.doi.org/10.1145/1250790.1250879}
  {\path{doi:10.1145/1250790.1250879}}.

\bibitem[CGI{\etalchar{+}}16]{carmosino2016nondeterministic}
M.~L. Carmosino, J.~Gao, R.~Impagliazzo, I.~Mihajlin, R.~Paturi, and
  S.~Schneider.
\newblock Nondeterministic extensions of the strong exponential time hypothesis
  and consequences for non-reducibility.
\newblock In {\em Proceedings of the 2016 ACM Conference on Innovations in
  Theoretical Computer Science}, ITCS '16, pages 261--270. ACM, 2016.
\newblock \href {http://dx.doi.org/10.1145/2840728.2840746}
  {\path{doi:10.1145/2840728.2840746}}.

\bibitem[CH03]{Cole03}
R.~Cole and R.~Hariharan.
\newblock A fast algorithm for computing steiner edge connectivity.
\newblock In {\em Proceedings of the Thirty-fifth Annual ACM Symposium on
  Theory of Computing}, STOC '03, pages 167--176. ACM, 2003.
\newblock \href {http://dx.doi.org/10.1145/780542.780568}
  {\path{doi:10.1145/780542.780568}}.

\bibitem[Chi60]{Chien60}
R.~T. Chien.
\newblock Synthesis of a communication net.
\newblock {\em IBM Journal of Research and Development}, 4(3):311--320, 1960.

\bibitem[CLL13]{CLL13}
H.~Y. Cheung, L.~C. Lau, and K.~M. Leung.
\newblock Graph connectivities, network coding, and expander graphs.
\newblock {\em SIAM Journal on Computing}, 42(3):733--751, 2013.
\newblock \href {http://dx.doi.org/10.1137/110844970}
  {\path{doi:10.1137/110844970}}.

\bibitem[Edm70]{Edmonds70}
J.~Edmonds.
\newblock Submodular functions, matroids, and certain polyhedra.
\newblock {\em Combinatorial structures and their applications}, pages 69--87,
  1970.

\bibitem[FF56]{FF56}
L.~R. Ford and D.~R. Fulkerson.
\newblock Maximal flow through a network.
\newblock {\em Canadian journal of Mathematics}, 8(3):399--404, 1956.
\newblock Available from: \url{http://www.rand.org/pubs/papers/P605/}.

\bibitem[Gab95]{Gabow95}
H.~N. Gabow.
\newblock A matroid approach to finding edge connectivity and packing
  arborescences.
\newblock {\em J. Comput. Syst. Sci.}, 50(2):259--273, 1995.

\bibitem[GGI{\etalchar{+}}17]{G+17}
L.~Georgiadis, D.~Graf, G.~F. Italiano, N.~Parotsidis, and P.~Uznanski.
\newblock {All-Pairs 2-Reachability in $O(n^\omega\log n)$ Time}.
\newblock In {\em 44th International Colloquium on Automata, Languages, and
  Programming (ICALP 2017)}, volume~80 of {\em Leibniz International
  Proceedings in Informatics (LIPIcs)}, pages 74:1--74:14. Schloss
  Dagstuhl--Leibniz-Zentrum fuer Informatik, 2017.
\newblock \href {http://dx.doi.org/10.4230/LIPIcs.ICALP.2017.74}
  {\path{doi:10.4230/LIPIcs.ICALP.2017.74}}.

\bibitem[GH61]{GH61}
R.~E. Gomory and T.~C. Hu.
\newblock Multi-terminal network flows.
\newblock {\em Journal of the Society for Industrial and Applied Mathematics},
  9:551--570, 1961.
\newblock Available from: \url{http://www.jstor.org/stable/2098881}.

\bibitem[GH86]{GH86}
F.~Granot and R.~Hassin.
\newblock Multi-terminal maximum flows in node-capacitated networks.
\newblock {\em Discrete Applied Mathematics}, 13(2-3):157--163, 1986.

\bibitem[GIKW17]{GIKW17}
J.~Gao, R.~Impagliazzo, A.~Kolokolova, and R.~R. Williams.
\newblock Completeness for first-order properties on sparse structures with
  algorithmic applications.
\newblock In {\em 28th Annual {ACM-SIAM} Symposium on Discrete Algorithms,
  {SODA} 2017}, pages 2162--2181, 2017.
\newblock \href {http://dx.doi.org/10.1137/1.9781611974782.141}
  {\path{doi:10.1137/1.9781611974782.141}}.

\bibitem[GT01]{GT01}
A.~V. Goldberg and K.~Tsioutsiouliklis.
\newblock Cut tree algorithms: an experimental study.
\newblock {\em Journal of Algorithms}, 38(1):51--83, 2001.

\bibitem[Gus90]{Gusfield90}
D.~Gusfield.
\newblock Very simple methods for all pairs network flow analysis.
\newblock {\em SIAM Journal on Computing}, 19(1):143--155, 1990.

\bibitem[HKP07]{HKP07}
R.~Hariharan, T.~Kavitha, and D.~Panigrahi.
\newblock Efficient algorithms for computing all low $s-t$ edge connectivities
  and related problems.
\newblock In {\em Proceedings of the 18th Annual ACM-SIAM Symposium on Discrete
  Algorithms}, pages 127--136. SIAM, 2007.
\newblock Available from:
  \url{http://dl.acm.org/citation.cfm?id=1283383.1283398}.

\bibitem[HL07]{HL07}
R.~Hassin and A.~Levin.
\newblock Flow trees for vertex-capacitated networks.
\newblock {\em Discrete Appl. Math.}, 155(4):572--578, 2007.
\newblock \href {http://dx.doi.org/10.1016/j.dam.2006.08.012}
  {\path{doi:10.1016/j.dam.2006.08.012}}.

\bibitem[IP01]{ImpaSETH}
R.~Impagliazzo and R.~Paturi.
\newblock On the complexity of k-{SAT}.
\newblock {\em Journal of Computer and System Sciences}, 62(2):367--375, March
  2001.
\newblock \href {http://dx.doi.org/10.1006/jcss.2000.1727}
  {\path{doi:10.1006/jcss.2000.1727}}.

\bibitem[Jel63]{Jelinek63}
F.~Jelinek.
\newblock On the maximum number of different entries in the terminal capacity
  matrix of oriented communication nets.
\newblock {\em IEEE Transactions on Circuit Theory}, 10(2):307--308, 1963.
\newblock \href {http://dx.doi.org/10.1109/TCT.1963.1082149}
  {\path{doi:10.1109/TCT.1963.1082149}}.

\bibitem[KT18]{KT18}
R.~Krauthgamer and O.~Trabelsi.
\newblock Conditional lower bounds for all-pairs max-flow.
\newblock {\em {ACM} Trans. Algorithms}, 14(4):42:1--42:15, 2018.
\newblock \href {http://dx.doi.org/10.1145/3212510}
  {\path{doi:10.1145/3212510}}.

\bibitem[K{\"u}n18]{Kunnemann18}
M.~K{\"u}nnemann.
\newblock On nondeterministic derandomization of freivalds' algorithm:
  Consequences, avenues and algorithmic progress.
\newblock {\em arXiv preprint arXiv:1806.09189}, 2018.

\bibitem[LNSW12]{LNSW12}
J.~Lacki, Y.~Nussbaum, P.~Sankowski, and C.~{Wulff-Nilsen}.
\newblock Single source - all sinks {M}ax {F}lows in planar digraphs.
\newblock In {\em Proc.\ of the 53rd FOCS}, pages 599--608, 2012.

\bibitem[LS14]{LS14}
Y.~T. Lee and A.~Sidford.
\newblock Path finding methods for linear programming: Solving linear programs
  in {\~{o}}(vrank) iterations and faster algorithms for {M}aximum {F}low.
\newblock In {\em Proc.\ of the 55th FOCS}, pages 424--433, 2014.

\bibitem[Mad16]{Madry16}
A.~Madry.
\newblock Computing maximum flow with augmenting electrical flows.
\newblock In {\em {IEEE} 57th Annual Symposium on Foundations of Computer
  Science, {FOCS} 2016}, pages 593--602, 2016.
\newblock \href {http://dx.doi.org/10.1109/FOCS.2016.70}
  {\path{doi:10.1109/FOCS.2016.70}}.

\bibitem[May60]{Mayeda60}
W.~Mayeda.
\newblock Terminal and branch capacity matrices of a communication net.
\newblock {\em IRE Transactions on Circuit Theory}, 7(3):261--269, 1960.
\newblock \href {http://dx.doi.org/10.1109/TCT.1960.1086673}
  {\path{doi:10.1109/TCT.1960.1086673}}.

\bibitem[May62]{Mayeda62}
W.~Mayeda.
\newblock On oriented communication nets.
\newblock {\em IRE Transactions on Circuit Theory}, 9(3):261--267, 1962.
\newblock \href {http://dx.doi.org/10.1109/TCT.1962.1086912}
  {\path{doi:10.1109/TCT.1962.1086912}}.

\bibitem[MMNS11]{certify_survey11}
R.~M. McConnell, K.~Mehlhorn, S.~N{\"a}her, and P.~Schweitzer.
\newblock Certifying algorithms.
\newblock {\em Computer Science Review}, 5(2):119--161, 2011.

\bibitem[Pan16]{Panigrahi16}
D.~Panigrahi.
\newblock {G}omory-{H}u trees.
\newblock In M.-Y. Kao, editor, {\em Encyclopedia of Algorithms}, pages
  858--861. Springer New York, 2016.
\newblock \href {http://dx.doi.org/10.1007/978-1-4939-2864-4_168}
  {\path{doi:10.1007/978-1-4939-2864-4_168}}.

\bibitem[ST18]{SidfordT18}
A.~Sidford and K.~Tian.
\newblock Coordinate methods for accelerating $\ell_\infty$ regression and
  faster approximate maximum flow.
\newblock In {\em {FOCS '18}}, pages 922--933. {IEEE} Computer Society, 2018.
\newblock \href {http://dx.doi.org/10.1109/FOCS.2018.00091}
  {\path{doi:10.1109/FOCS.2018.00091}}.

\bibitem[Wil05]{Wil05}
R.~Williams.
\newblock A new algorithm for optimal 2-constraint satisfaction and its
  implications.
\newblock {\em Theor. Comput. Sci.}, 348(2-3):357--365, 2005.
\newblock \href {http://dx.doi.org/10.1016/j.tcs.2005.09.023}
  {\path{doi:10.1016/j.tcs.2005.09.023}}.

\bibitem[Wil16]{Wil16_MASETH}
R.~R. Williams.
\newblock Strong {ETH} breaks with {M}erlin and {A}rthur: Short non-interactive
  proofs of batch evaluation.
\newblock In {\em 31st Conference on Computational Complexity, {CCC} 2016},
  pages 2:1--2:17, 2016.

\bibitem[WL93]{WL93}
Z.~Wu and R.~Leahy.
\newblock An optimal graph theoretic approach to data clustering: Theory and
  its application to image segmentation.
\newblock {\em IEEE transactions on pattern analysis and machine intelligence},
  15(11):1101--1113, 1993.

\end{thebibliography}
}

\end{document}